\documentclass[11pt]{article}

\usepackage{fullpage}
\usepackage{amssymb,amsmath}
\usepackage{graphicx, epsfig}

\def\idrm#1{\ensuremath{\mathrm{#1}}}
\def\idtt#1{\ensuremath{\mathtt{#1}}}
\def\idsf#1{\ensuremath{\mathsf{#1}}}

\def\floor#1{\lfloor #1 \rfloor}
\def\ceil#1{\lceil #1 \rceil}

\newcommand{\no}[1]{}

\newtheorem{theorem}{Theorem}
\newtheorem{lemma}{Lemma}

\newenvironment{proof}{\trivlist\item[]\emph{Proof}:}%
{\unskip\nobreak\hskip 1em plus 1fil\nobreak$\Box$
\parfillskip=0pt%
\endtrivlist}

\newenvironment{itemize*}%
  {\begin{itemize}%
    \setlength{\itemsep}{0pt}%
    \setlength{\parskip}{0pt}%
    \setlength{\parsep}{0pt}%
    \setlength{\topsep}{0pt}%
    \setlength{\partopsep}{0pt}%
  }%
  {\end{itemize}}%

\newcommand{\cT}{{\cal T}}

\newcommand{\cQ}{{\cal Q}}

\newcommand{\cN}{{\cal N}}

\newcommand{\tS}{\widetilde{S}}
\newcommand{\tI}{\widetilde{I}}

\newcommand{\oE}{\overline{E}}

\newcommand{\ptr}{\idsf{ptr}}
\newcommand{\ppar}{\idtt{par}}
\newcommand{\mytop}{\idsf{top}}
\newcommand{\rank}{\idrm{rank}}
\newcommand{\pred}{\idrm{pred}}
\newcommand{\ssucc}{\idrm{succ}}
\newcommand{\docc}{\mathit{docc}}
\newcommand{\occ}{\mathit{occ}}

\newcommand{\mindist}{\mathit{mindist}}
\newcommand{\docrank}{\mathit{docrank}}
\newcommand{\tf}{\mathit{tf}}
\newcommand{\idf}{\mathit{idf}}

\newcommand{\Yx}{\mathit{Yx}}
\newcommand{\Yw}{\mathit{Yw}}

\newcommand{\eps}{\varepsilon}

\begin{document}

\no{\title{\bf Optimal Top-$k$ Document Retrieval \\ and Ranked Range Reporting 
in Linear Space}}
\title{Optimal Top-$k$ Document Retrieval%
\thanks{Partially funded by Fondecyt Grant 1-110066, Chile, and by
Millennium Nucleus Information and Coordination in Networks ICM/FIC P10-024F,
Chile. An early partial version of this article appeared in {\em Proc. SODA
2012} \cite{NN12}.}} 
\author{
\begin{tabular}{cc}
Gonzalo Navarro &
 Yakov Nekrich \\ 
 Department of Computer Science~~~~ &
 ~~~~Department of Computer Science \\
 University of Chile, Chile &
 University of Kansas, USA Ã‚Â¸\\
{\tt gnavarro@dcc.uchile.cl} &
{\tt yakov.nekrich@googlemail.com} 
\end{tabular}}
\date{}

\maketitle

\begin{abstract}
 Let $\mathcal D$ be a collection of $D$ documents, which are
strings over an alphabet of size $\sigma$, of total length $n$.
We describe a data structure that uses linear space and 
and reports $k$ most relevant documents that contain a query
pattern $P$, which is a string of length $p$, in time $O(p/\log_\sigma n+k)$,
which is optimal in the RAM model in the general case where $\lg D =
\Theta(\log n)$, and involves a novel RAM-optimal suffix tree search. 
Our construction supports an ample set of important relevance 
measures, such as the number of times $P$ appears in a document (called term
frequency), a fixed document importance, and the minimal distance 
between two occurrences of $P$ in a document. 

When $\lg D = o(\log n)$, we show how to reduce the space 
of the data structure from $O(n\log n)$ to $O(n(\log\sigma+\log D+\log\log n))$
bits, and to $O(n(\log\sigma+\log D))$ bits in the case of the popular 
term frequency measure of relevance, at the price of an additive term
$O(\log^\eps n \log\sigma)$ in the query time, for any constant $\eps>0$.

We also consider the dynamic scenario, where documents
can be inserted and deleted from the collection. We obtain linear space and
query time $O(p(\log\log n)^2/\log_\sigma n+\log n + k\log\log k)$, 
whereas insertions and deletions require $O(\log^{1+\eps} n)$ time per symbol,
for any constant $\eps>0$. 

Finally, we consider an extended static scenario where an extra parameter 
$\ppar(P,d)$ is defined, and the query must retrieve only documents $d$
such that $\ppar(P,d)\in [\tau_1,\tau_2]$, where this range is specified at 
query time. We solve these queries using linear space and 
$O(p/\log_\sigma n + \log^{1+\eps} n + k\log^\eps n)$ time, for any constant 
$\eps>0$.

Our technique is to translate these top-$k$ problems into multidimensional 
geometric search problems. As an additional bonus, we describe some 
improvements to those problems.
\end{abstract}

\section{Introduction}

The design of efficient data structures for document (i.e., string) collections
that can report those containing a query pattern $P$
is an important problem  studied in the 
information retrieval and pattern matching communities
(see, e.g., a recent survey \cite{Nav13}). 
Due to the steadily increasing volumes of data, 
it is often necessary to generate a list $L(P)$ of the documents 
containing a string pattern $P$ in decreasing order of
relevance. Since the list $L(P)$
can be very large, in most cases we are interested in 
answering \emph{top-$k$ queries}, that is, 
reporting only the first $k$ documents from $L(P)$
for a parameter $k$ given at query time.

Inverted files~\cite{BBHME87,Knu73,BYRN11} 
that store lists of documents containing 
certain keywords 
are frequently used in practical implementations of 
information retrieval methods. 
However, inverted files only work when query patterns 
belong to  a fixed pre-defined set of strings (keywords). 
The suffix tree~\cite{Wei73}, a handbook data structure known since 1973, 
uses linear space (i.e., $O(n)$ words, where $n$ is the total length of all the
documents) and finds all the $\occ$
occurrences of a pattern $P$ in $O(p+ \occ)$ time, where $p=|P|$.
Surprisingly, the general \emph{document listing problem}, 
that is, the problem of reporting all the documents that contain 
an arbitrary query pattern $P$, was not studied until 
the end of the 90s.  
Suffix trees and other data structures for standard 
pattern matching queries do not provide a satisfactory 
solution for the document listing problem because 
the same document may contain many occurrences of $P$.
Matias et al.~\cite{MMSZ98} described the first data structure 
for  document listing queries; their structure 
uses $O(n)$ words of  space and reports all $\docc$ documents that 
contain $P$  in $O(p\log D+ \docc)$ time, where $D$ is the total number of
documents in the collection. 
Muthukrishnan~\cite{Mut02} presented a data structure that 
uses $O(n)$ words of space and answers document listing queries 
in $O(p+\docc)$ time. 
Muthukrishnan~\cite{Mut02} also initiated the study of more sophisticated 
problems in which only  documents that contain $P$ and satisfy some 
further criteria are reported. 
In the $K$-mining problem, we must report documents in which $P$ occurs
 at least 
$K$ times; in the $K$-repeats problem, we must report  documents in which 
at least two occurrences of $P$ are within  a distance $K$.  
He described $O(n)$- and an $O(n\log n)$-word data structures that answer 
$K$-mine and $K$-repeats queries, respectively, both in $O(p+\occ)$ time, 
where $\occ$ is the number of reported documents.

A problem not addressed by Muthukrishnan, and arguably the most important one 
for information retrieval, is the  {\em top-$k$ document retrieval}
problem:  report   $k$ most highly ranked documents for a query pattern
 $P$ in decreasing order of their ranks. The ranking is 
measured with respect to the so-called {\em relevance} of a string $P$ for a 
document $d$. A basic relevance measure is $\tf(P,d)$, the number of 
times $P$ occurs in $d$. Two other important examples 
are  $\mindist(P,d)$, the minimum distance between two occurrences of 
$P$ in $d$, and $\docrank(d)$, an arbitrary static rank assigned to a document $d$. 
Some more complex measures have also been proposed. 
Hon et al.~\cite{HPSW10} presented a solution for the top-$k$ document 
retrieval problem for the case when the relevance measure is $\tf(P,d)$. 
Their data structure uses $O(n\log n)$ words of space and answers queries 
in $O(p+k+\log n \log\log n)$ time. Later,
Hon, Shah and Vitter \cite{HSV09} presented a general solution 
for a wide class of relevance measures. 
Their data structure uses
 linear space and needs $O(p+ k\log k)$ time to answer a 
top-$k$ query. A recent $O(n)$ space data structure 
\cite{KN11} enables us to answer top-$k$ queries in $O(p+k)$ time when
the relevance measure is $\docrank(d)$. 
However, that result cannot be extended to other
more important  relevance measures.

\paragraph{Our Results.}
Hon et al.'s results~\cite{HSV09} are an important achievement, but their 
time is not yet
optimal. In this paper we describe a linear space data structure that answers
top-$k$ document queries in $O(p/\log_\sigma n+k)$ time, where $\sigma$ is
the alphabet size of the collection. This is optimal in
the $\Theta(\log n)$-word RAM model we use, unless the collection has very 
few documents, $\lg D = o(\log n)$ (i.e., $D=o(n^\eps)$ for any
constant $\eps>0$). We support the same relevance measures as Hon 
et al.~\cite{HSV09}. 

\begin{theorem}\label{theor:topk}
  Let $\mathcal{D}$ be a collection of strings (called {\em documents}) of 
  total length $n$ over an integer alphabet $[1,\sigma]$, and let $w(S,d)$ be a 
  function  that assigns a numeric {\em weight}
  to string $S$ in document $d$, so that $w(S,d)$ depends only on the set of 
  starting positions of occurrences of $S$ in $d$.
  Then there exists an $O(n)$-word space data structure that, given a string
  $P$ of length $p$ and an integer $k$, reports $k$ documents $d$ containing 
  $P$ with highest $w(P,d)$ values, in decreasing order of $w(P,d)$, in 
  $O(p/\log_\sigma n+k)$ time. The time is online on $k$. 
\end{theorem}

Note that the weighting function is general enough to encompass measures 
$\tf(P,d)$, $\mindist(P,d)$ and $\docrank(d)$. 
As stated, our solution is {\em online} on $k$:
It is not necessary to specify $k$ beforehand; our
data structure can simply report documents in decreasing 
relevance order until all the documents are reported or the query processing 
is terminated by a user.

An online top-$k$ solution using the $\tf$ measure solves the
$K$-mining problem in optimal time and linear space. An online top-$k$
solution using $\mindist$ measure solves the $K$-repeats problem in optimal
time and linear space. We remind that Muthukrishnan \cite{Mut02} had solved 
the $K$-repeats problem using $O(n\log n)$-word space; later Hon et 
al.~\cite{HSV09} reduced the space to linear. Now all these results appear
as a natural corollary of our optimal top-$k$ retrieval solution. Our results 
also subsume those on more recent variants of the problem \cite{KN11}, for
example when the rank $\docrank(d)$ depends only on $d$ (we just use
$w(P,d) = \docrank(d)$), or where in addition
we exclude those $d$ where $P$ appears less than $K$ times for a fixed 
pre-defined $K$ (we just use $w(P,d) = \docrank(d)$ if $\tf(P,d) \ge K$, 
else 0). 

Moreover, we can also answer queries for some relevance metrics not 
included in Theorem~\ref{theor:topk}. For instance, we might be interested 
in reporting all the documents $d$ with $\tf(P,d) \times \idf(P)
\ge \tau$, where $\idf(P) = \log (N/\mathit{df}(P))$ and 
$\mathit{df}(P)$ is the number of documents where $P$ appears \cite{BYRN11}.
Using the $O(n)$-bit structure of Sadakane~\cite{Sad07}, we can compute 
$\idf(P)$ in $O(1)$ time from the suffix tree locus of $P$. To answer the query,
 we use our data structure 
of Theorem~\ref{theor:topk} in online mode on measure $\tf$: For every 
reported document $d$ we find $\tf(P,d)$ and 
compute  $\tf(P,d) \times \idf(P)$; the procedure is terminated 
when a document $d_l$ with $\tf(P,d_l)\times \idf(P)<\tau$ 
is encountered. 
Thus we need  $O(p/\log_{\sigma} n+\occ)$ time to report all
 $\occ$ documents with $\tf \times \idf$ scores above a threshold.

\medskip

When $\lg D = o(\log n)$, it is not clear that our time is RAM-optimal. Instead,
we show that in this case the space of our data structures can be reduced from 
$O(n\log n)$ bits to $O(n(\log\sigma + \log D + \log\log n))$. This is 
$o(n\log n)$ bits unless $\lg\sigma = \Theta(\log n)$ (in which case the
linear-space data structure is already asymptotically optimal).
For the most important $\tf$ relevance 
measure, where we report documents in which $P$ occurs most frequently, 
we obtain a data structure that uses $O(n(\log\sigma+\log D))$ bits of space.  
The price of the space reduction is an additive term $O(\log^\eps n\log\sigma)$
in the query time, for any constant $\eps>0$.

\medskip

We also consider the dynamic framework, where collection $\mathcal{D}$
admits insertions of new documents and deletions of existing documents. Those
updates are supported in slightly superlogarithmic time per character, whereas 
the query times are only slightly slowed down. We note that measure $C_w$ is 
just $O(1)$ for the typical relevance measures $\tf$ and $\docrank$, and
$O(\log n)$ for $\mindist$.

\begin{theorem}\label{theor:dyntopk}
  Let $\mathcal{D}$ be a collection of documents of 
  total length $n$ over an integer alphabet $[1,\sigma]$, and let $w(S,d)$ be 
  a function  that assigns a numeric {\em weight}
  to string $S$ in document $d$, so that $w(S,d)$ depends only on the set of 
  starting positions of occurrences of $S$ in $d$, and can be computed in
  $O(C_w|d|)$ time for all the nodes of the suffix tree of document $d$.
  Then there exists an $O(n)$-word space data structure that, given a string
  $P$ of length $p$ and an integer $k$, reports $k$ documents $d$ containing 
  $P$ with highest $w(P,d)$ values, in decreasing order of $w(P,d)$, in 
  $O(p(\log\log n)^2/\log_\sigma n+\log n+k\log\log k)$ time, online in $k$.
  The structure can insert new documents and delete existing documents in 
  $O(C_w+\log^{1+\eps} n)$ time per inserted character and
  $O(\log^{1+\eps} n)$ per deleted character, for any constant 
  $\eps>0$.
\end{theorem}

We note that a direct
dynamic implementation of the solution of Hon et al.~\cite{HSV09} would 
require at least performing $p+k$ dynamic RMQs, which cost
$\Omega(\log n / \log\log n)$ time \cite{AHR98}. Thus modeling the original
problem as a geometric one pays off in the dynamic scenario as well.

\medskip

Furthermore, we can extend the top-$k$ ranked retrieval problem by allowing
a further parameter $\ppar(P,d)$ to be associated to any pattern $P$ and
document $d$, so that only documents with $\ppar(P,d)\in [\tau_1,\tau_2]$ are
considered. 
Some applications are selecting a range of creation dates,
lengths, or PageRank values for the documents (these do not depend on $P$),
bounding the allowed number of occurrences of $P$ in $d$, or the minimum
distance between two occurrences of $P$ in $d$, etc.

\begin{theorem}\label{theor:partopk}
  Let $\mathcal{D}$ be a collection of  documents of 
  total length $n$ over an integer alphabet $[1,\sigma]$, 
  let $w(S,d)$ be a function  that assigns a numeric {\em weight}
  to string $S$ in document $d$, and let $\ppar(S,d)$ be another parameter, 
  so that $w$ and $\ppar$ depend only on the set of starting positions of 
  occurrences of $S$ in $d$.
  Then there exists an $O(n)$-word space data structure that, given a string
  $P$ of length $p$, an integer $k$, and a range $[\tau_1,\tau_2]$, reports $k$ documents $d$ 
  containing $P$ and with $\ppar(P,d) \in [\tau_1,\tau_2]$, with highest $w(P,d)$ 
  values, in decreasing order of $w(P,d)$, in 
  $O(p/\log_\sigma n+\log^{1+\eps}n + k\log^\eps n)$ time, online in $k$,
  for any constant $\eps>0$.
\end{theorem}

Our solutions map these document retrieval problems into range search problems
on multidimensional spaces, where points in the grids have associated weights.
We improve some of the existing solutions for those problems.

An early partial version of this article appeared in {\em Proc. SODA 2012}
\cite{NN12}. This extended version includes, apart from more precise 
explanations and fixes, the improvement of the static results to achieve
RAM-optimality on the suffix tree traversal, and the new results on the dynamic
scenario. The paper is organized as follows. In Section~\ref{sec:topkfram} we
review the top-$k$ framework of Hon et al.~\cite{HSV09} and reinterpret it as
the combination of a suffix tree search plus a geometric search problem. We
introduce a RAM-optimal suffix tree traversal technique that is of independent
interest, and state our results on geometric grids, each of which is related 
to the results we achieve on document retrieval. Those can also be of 
independent interest. Sections~\ref{sec:polyn} and \ref{sec:optimal} describe
our basic static solution. Section~\ref{sec:dynamic} describes our dynamic
solution. In Section~\ref{sec:freqspace} we show how the static solution can
be modified to reduce its space requirements, and in
Section~\ref{sec:topk2dim} we show how it can be extended to support an
additional restriction on the documents sought. Finally,
Section~\ref{sec:concl} concludes and gives future work directions.

\section{Top-$k$ Framework}
\label{sec:topkfram}

In this section we overview the framework of Hon, Shah, and
Vitter~\cite{HSV09}.  Then, we describe a geometric interpretation 
of  their structure and show how
top-$k$ queries can be reduced to a special case of range reporting queries
on a grid.

Let $T$ be the generalized suffix tree \cite{Wei73,McC76,Ukk95} for a 
collection of 
documents $d_1,\ldots,d_D$, each ending with the special terminator symbol 
``\$''. $T$  is a compact trie, such that all suffixes of all documents are 
stored in the leaves of $T$. 
We denote by $path(v)$ the string obtained by concatenating the labels 
of all the edges on the path from the root to $v$. 
The \emph{locus} of a string $P$ is the highest node $v$ such that 
$P$ is a prefix of $path(v)$. Every occurrence of $P$ corresponds to 
a unique leaf that descends from its locus. We refer the reader to classical
books and surveys \cite{Apo85,Gus97,MR02} for an extensive description of this 
data structure.

We say that a leaf $l$ is {\em marked} with document $d$ if the 
suffix stored in $l$ belongs to $d$. An internal node $v$ is marked with $d$
if at least two children of $v$ contain leaves marked with $d$.  While
a leaf is marked with only one value $d$ (equal suffixes of distinct documents
are distinguished by ordering the string terminators arbitrarily), an internal 
node can be marked with many values $d$.

In every node $v$ of $T$ marked with $d$, we store a pointer
$\ptr(v,d)$ to its lowest ancestor $u$ such that $u$ is also marked
with $d$.  If no ancestor $u$ of $v$ is marked with $d$, then
$\ptr(v,d)$ points to a dummy node $\nu$ such that $\nu$ is the parent
of the root of $T$.  We also assign a weight to every pointer $\ptr(v,d)$. 
This weight 
is the relevance score of the document $d$ with
respect to the string $path(v)$. 
The following statements hold; we reprove them for completeness.

\begin{lemma}[\cite{HSV09}, Lemma 4]
  The total number of pointers $\ptr(\cdot,\cdot)$ in $T$ is bounded by $O(n)$.
\end{lemma}
\begin{proof}
  The total number of pointers $\ptr(v,d)$, $v\in T$, does not exceed
  the number of nodes marked with $d$.  The total number of internal
  nodes marked with $d$ is smaller than the number of leaves marked
  with $d$.  Since there are $O(|d|)$ leaves marked with $d$, the
  total number of pointers $\ptr(v,d)$ for a fixed document $d$ is
  bounded by $O(|d|)$, and those $|d|$ add up to $n$.
\end{proof}

\begin{lemma}[\cite{HSV09}, Lemma 2] \label{lemma:uniqptr}
  Assume that document $d$ contains a pattern $P$ and $v$ is the
  locus of $P$. Then there exists a unique pointer $\ptr(u,d)$, such
  that $u$ is in the subtree of $v$ (which includes $v$) and $\ptr(u,d)$ points to an
  ancestor of $v$. 
\end{lemma}
\begin{proof}
  If $d$ contains $P$ then there is at least one leaf $u$ marked $d$ below the 
locus of $P$, with a pointer $\ptr(u,d)$. If there are two maximal (in the
sense of ancestorship) nodes $u$ and $u'$ below $v$ with pointers $\ptr(u,d)$
and $\ptr(u',d)$, then their lowest common ancestor $v'$ is also
marked. Since $v$ is an ancestor of $u$ and $u'$, $v$ is $v'$ or an ancestor 
of $v'$ and then $\ptr(u,d)$ and $\ptr(u',d)$ must point to $v'$, not to an 
ancestor of $v$. Finally, if $u$ is a unique maximal node with $\ptr(u,d)$
(not $u$ might be $v$), then it must point an ancestor of $v$.
\end{proof}

Moreover, in terms of Lemma~\ref{lemma:uniqptr}, it turns out that
$path(u)$ occurs in $d$ at the same positions as $path(v)$.
 Note that the starting positions
of $P$ and of $path(v)$ in $d$ are the same, since $v$ is the locus of $P$,
and those are the same as the starting positions  of $path(u)$ in $d$. 
Thus $w(P,d)=w(path(u),d)$ for any measure $w(\cdot,\cdot)$ considered
in Theorem~\ref{theor:topk}. 

\paragraph{RAM-optimal suffix tree traversal.} To achieve time $O(p)$ for 
the locus search in the suffix tree while retaining linear space, one needs 
to organize the children of each node in a perfect hash function (phf)
\cite{FKS84}.
In order to reduce this time to the RAM-optimal $O(p/\log_\sigma n)$, we
proceed as follows. Let $l(u,v)$ be the concatenation of string labels from
node $u$ to its descendant $v$. We collect in a phf $H(u)$,
for the suffix tree root $u$, all the highest descendants $v$ such that
$|l(u,v)| \ge \ell = \log_\sigma n$. Those nodes $v$ are indexed with a
key built from the first $\ell$ symbols of $l(u,v)$ interpreted as a number
of $\lg n$ bits. 
We build recursively phfs for
all the descendants $u$ identified. Since each suffix tree node is included
in at most one hash table, the total size is $O(n)$ and the total deterministic
construction time is $O(n\log\log n)$ \cite{Ruz08}. 
Now $P$ is searched for as follows.
We take its first chunk of $\ell$ characters, interpret it as a number, and
query the phf of the root. If no node $u$ is found for that prefix of
$P$, then $P$ is not in the collection. Otherwise, the string depth of $u$ is
$\ge \ell$. We check explicitly the extra $|l(u,v)|-\ell$ symbols in the text,
by comparing chunks of $\ell$ symbols. If there is a match, we continue with
the next $\ell$ unread symbols of $P$, and so on, until there are less than
$\ell$ symbols to match in the remaining suffix $P'$ of $P$. At this point we 
switch to using a {\em weak prefix search (wps)} data structure \cite{BBPV10}: 
For each node $u$ holding a phf, we also store this wps data 
structure with all the nodes $v$ that descend from $u$ where $|l(u,v)| < \ell$.
The wps data structure will return the lexicographic range of the nodes of 
which $P'$ is a prefix (here we identify nodes $v$ with the strings $l(u,v)$).
The first in that range is the locus of $P$. One detail is that, if there is
no such node (i.e., $P$ has no locus), the wps structure returns an arbitrary
value, but this can be easily checked in optimal time in the text. The wps
structure requires, for our length $|P'| < \ell$ and in the RAM model, $O(1)$
query time, and can use, for the strings we store, of length $< \ell$, 
$O(\sqrt{\log n} \log\log n)$ bits, or $o(1)$ words, per stored node
\cite[Thm.~6]{BBPV10}. 
Once again, each node
is stored only in one wps structure, so the overall extra space is linear at
worst. The wps construction is $O(n\log^\eps n)$ randomized time. It can 
be made deterministic $O(n\,\textrm{polylog}(n))$ time by using a phf inside 
the construction \cite{AN96}. By replacing the wps structure by layered phfs
for $(\lg_\sigma n)/2^i$ symbols, we would have an additive term 
$O(\log\log_\sigma n)$ in the query time.

\paragraph{Geometric interpretation.}  We index the nodes of $T$ in the
following way: All nodes of $T$ are visited in pre-order; we also initialize
an index $i \leftarrow 0$. When a node $v$ is visited, if $v$
is marked with values $d_{v_1},\ldots,d_{v_j}$, we assign indexes
$i+1,\ldots,i+j$ to $v$ and set $i \leftarrow i+j$.  
We will denote by $[l_v,r_v]$ the integer
interval bounded by the minimal and maximal indexes assigned to $v$ or its
descendants.  Values $l_v$ and $r_v$ are stored in node $v$
of $T$.  Furthermore, for every $d_{v_t}$, $1\leq t\leq j$,
there is a pointer $\ptr(v,d_{v_t})$ that points to some ancestor
$u_t$ of $v$.  We encode $\ptr(v,d_{v_t})$ as a point $(i+t,depth(u_t))$,
where $depth$ denotes the depth of a node; $depth(\nu)=0$.
Thus every pointer in $T$ is encoded as a two-dimensional point on an
integer $O(n)\times O(n)$ grid. The weight of a point $p$ is that of the
pointer $p$ encodes.  We observe that all the points have different
$x$-coordinates.  Thus we obtain a set $S$ of weighted points with
different $x$-coordinates, and each point corresponds to a unique
 pointer.

For the final answers we will need to convert the $x$-coordinates of points 
found on this grid into document numbers. We store a global array of size
$O(n)$ to do this mapping.

\paragraph{Answering queries.}  Assume that top-$k$ documents containing a
pattern $P$ must be reported. We find the locus $v$ of $P$ in $O(p/\log_\sigma n)$ 
time.  By Lemma~\ref{lemma:uniqptr}, there is a unique pointer
$\ptr(u,d)$, such that $u$ is a descendant of $v$ (or $v$ itself) and $\ptr(u,d)$
points to an ancestor of $v$, for every document $d$ that contains
$P$. Moreover the weight of that point is $w(P,d)$. 
Hence, there is a unique point $(x,y)$ with $x\in [l_v,r_v]$ and
$y\in [0,depth(v)-1]$ for every document $d$ that contains $P$.  Therefore,
reporting top-$k$ documents is equivalent to the following query:
among all the points in the three-sided range $[l_v,r_v]\times [0,depth(v))$,
report $k$ points with highest weights.  We will call such queries
{\em three-sided top-$k$ queries}. In Sections~\ref{sec:polyn} 
and~\ref{sec:optimal} we prove the following result.
Theorem~\ref{theor:topk} is an immediate corollary of it,
as $h=depth(v)-1$ and $depth(v) \le p$, and we can choose $c \ge 1$.

\begin{theorem}\label{theor:weight}
  A set of $n$ weighted points on an $n\times n$ grid can be stored in
  $O(n)$ words of space and built in $O(n\log n)$ time, 
  so that for any $1\leq k,h\leq n$ and
  $1\leq a\leq b\leq n$, $k$ most highly weighted points in the range
  $[a,b]\times [0,h]$ can be reported in decreasing order of their weights
  in $O(h/\log^c n+k)$ time, for any constant $c$. 
\end{theorem}

\paragraph{Top-$k$ queries on dynamic collections.}
The static suffix tree is replaced by a dynamic one, with search time 
$O(p(\log\log n)^2/\log_\sigma n+\log n)$ 
and update time $O(\log n)$ per symbol. 
We must also update the grid, for which
we must carry out lowest common ancestor queries on the dynamic suffix tree and
also insert/delete points (and columns) in the grid. We split the grid of 
Theorem~\ref{theor:weight} into horizontal stripes of height 
$m=\mathrm{polylog}~n$ to obtain
improved performance, and query the highest $\lceil p/m\rceil$ of those 
grids. In the most general case (i.e., the last grid) we carry out a 
three-sided top-$k$ query. We address this in Section~\ref{sec:dynamic},
where in particular we prove the following result on dynamic grids.
Theorem~\ref{theor:dyntopk} is then obtained by combining those results.

\begin{theorem}\label{theor:dynweight}
  A set of $n$ points, one per column on an $n\times n$ grid, 
  with weights in $[1,O(n)]$, can be stored in
  $O(n)$ words of space, so that for any $1\leq k \le n$, $1 \le h\leq n$ and
  $1\leq a\leq b\leq n$, $k$ most highly weighted points in the range
  $[a,b]\times [0,h]$ can be reported in 
  decreasing order of their weights in $O(h/\log^c n + \log n + 
  k\log\log k)$ time, online in $k$, for any constant $c$. 
  Points (and their columns) can be inserted and deleted 
  in $O(\log^{1+\eps} n)$ time, for any constant $\eps>0$.
\end{theorem}

\paragraph{Parameterized top-$k$ queries.}
We use the same geometric interpretation as described above, but now
each pointer $\ptr(v,d)$ is also
associated with the parameter value $\ppar(path(v),d)$.  
We encode a pointer $\ptr(v,d_{v_t})$ as a three-dimensional 
point $(i+t,depth(u_t),\ppar(path(v),d_{v_t}))$, where $i$, $t$, and $u_t$ 
are defined as in the case of nonparameterized top-$k$ queries.
All the documents that contain a pattern $P$ (with locus $v$) and satisfy 
$\tau_1\leq \ppar(P,d)\leq \tau_2$ correspond to unique points 
in the range $[l_v,r_v]\times [0,depth(v))\times [\tau_1,\tau_2]$. 
Hence, reporting top-$k$ documents with $\ppar(P,d)\in [\tau_1,\tau_2]$ 
is equivalent to reporting top-$k$ points in a three-dimensional range. 
The following result is proved in Section~\ref{sec:topk2dim}, and
Theorem~\ref{theor:partopk} is an immediate corollary of it, choosing any
$c \ge 1$.

\begin{theorem}\label{theor:parweight}
  A set of $n$ weighted points on an $n\times n\times n$ grid can be stored in
  $O(n)$ words of space, so that for any $1\leq k,h\leq n$,
  $1\leq a\leq b\leq n$, and $1\leq \tau_1\leq \tau_2 \leq n$, 
  $k$ most highly weighted points in the range
  $[a,b]\times [0,h]\times [\tau_1,\tau_2]$ can be reported in 
  decreasing order of their weights in $O(h/\log^c n+\log^{1+\eps} n + 
  k\log^\eps n)$ time, for any constants $c$ and $\eps>0$.
\end{theorem}

\section{An $O(m^f+k)$ Time Data Structure}
\label{sec:polyn}

In this section we give a data structure that does not yet achieve the desired
$O(h/\log n+k)$ time, but its time depends on the width $m$ of the grid. This will be 
used in Section~\ref{sec:optimal} to handle vertical stripes of the global
grid, in order to achieve the final result.

We assume that a global array gives access to the points of a set $S$ in 
constant time: if we know the $x$-coordinate $p.x$ of a point $p\in S$,  we 
can obtain the $y$-coordinate $p.y$ of $p$ in $O(1)$ time. Both $p.x$ and 
$p.y$ are in $[1,O(n)]$, thus the global array requires $O(n)$ words of space. 
We consider the question of how much additional space our data structure uses 
if this global array is available. The result of this section is summed up in 
the following lemma, where we consider tall grids of $m$ columns and $n$ rows.

\begin{lemma}\label{lemma:polyn}
  Assume that  $m\leq n$ and let $0<f<1$ be a constant. There exists a data
  structure that uses $O(m\log m)$ additional bits of space and
  construction time. It
  answers three-sided top-$k$ queries for a set of $m$ points on an
  $m\times n$ grid in $O(m^f+k)$ time.
\end{lemma}
\begin{proof}
The idea is to partition the points of $S$ by weights, where the weights are
disregarded inside each partition. Those partitions are also refined into a tree.
Then we solve the problem by traversing the appropriate partitions and 
collecting all the points using classical range queries on unweighted points. 
A tree of arity $m^{\Theta(f)}$ yields constant height and thus constant space
 per point.

More precisely, we partition $S$ into classes $S_1,\ldots, S_r$, where
$r=m^{f'}$ for a constant $0<f'<f$.  For any $1\leq i< j\leq r$, the
weight of any point $p_i\in S_i$ is larger than the weight of any
point $p_j\in S_j$.  For $1\leq i<r$, $S_i$ contains $m^{1-f'}$
points.  Each class $S_i$ that contains more than one element is
recursively divided into $\min(|S_i|,r)$ subclasses in the same
manner.  This subdivision can be represented as a tree: If $S_i$ is
divided into subclasses $S_{i_1},\ldots, S_{i_k}$, we will say that
$S_i$ is the parent of $S_{i_1},\ldots, S_{i_k}$. This tree has constant 
height $O(1/f')$.

For every class $S_j$ we store data structures that support
three-sided range counting queries and three-sided range reporting
queries, in $O(\log m)$ and $O(\log m +occ)$ time respectively. These
structures will be described in Section~\ref{sec:wtree} and require
$O(m^{1-f'}\log m)$ construction time; 
note they do not involve weights. This adds up to $O((1/f')m\log m)$
construction time.

We will report
$k$ most highly weighted points in a three-sided query range
$Q=[a,b]\times [0,h]$ using a two-stage
procedure. During the first stage we produce an unsorted list $L$ of
$k$ most highly weighted points. During the second stage the list $L$
is sorted by weight.

Let $Q^k$ denote the set of $k$ most highly weighted points in $S\cap
Q$. Then $Q^k$ can be formed as the union of the result of the three-sided
query over certain classes, at most $O(m^{f'})$ clases per
level over a constant number of levels. More precisely, 
there are $O(m^{f'})$ classes $S_c$, such that $p\in Q^k$
if and only if $p\in S_c\cap Q$ for some $S_c$.  During the first
stage, we identify the classes $S_c$ and report all the points in $S_c\cap
Q$ using the following procedure.  Initially, we set our current tree node
to $\tS=S$ and its child number to $i=1$.
We count the number of points inside $Q$ in the $i$-th child
$S_i$ of $\tS$.  If $k_i=|S_i\cap Q| \leq k$, we report all the points
from $S_i\cap Q$ and set $k=k-k_i$.  If $k=0$, the procedure is
completed; otherwise we set $i=i+1$ and proceed to the next child
$S_i$ of $\tS$.  If $k_i> k$, instead, we set $\tS=S_i$, $i=1$, and 
report $k$ most highly weighted points in the children of $\tS$ using the same
procedure. During the first stage we examine $O(m^{f'})$ classes $S_i$
and spend $O(m^{f'}\log m + k)=O(m^f+k)$ time in total.

When the list $L$ is completed, we can sort it in $O(m^f+k)$ time. If
$L$ contains $k<m^{f'}$ points, $L$ can be sorted in $O(k \log k)
=O(m^f)$ time. If $L$ contains $k\geq m^{f'}$ points, then we can sort it in 
$O(k)$ time using radix sort: As the set $S$ contains at most $m$ distinct
weights, we store their ranks in an array ordered by $x$-coordinate, and thus
can sort the result using the ranks instead of the original values. By sorting 
$f'\log_2 m$ bits per pass the radix sort runs in time $O(k)$.

As for the space, the structures in Lemmas~\ref{lemma:count1} and
\ref{lemma:rep1} require $O(\log m)$ bits per point. Each point of $S$ belongs 
to $O(1/f')$ classes $S_i$. Hence, the total number of points in all classes 
is $O(m)$, giving $O(m\log m)$ bits of total space. The local
array of weight ranks also uses $O(m\log m)$ bits.
\end{proof}
 
\subsection{Counting and Reporting Points} \label{sec:wtree}

It remains to describe the data structures that answer three-sided
counting and reporting queries, with no weights involved.

\begin{lemma}\label{lemma:count1}
  Let $v\leq m\leq n$.  There exists a data structure that uses
  $O(v\log m)$ additional bits and answers three-sided range counting
  queries for a set of $v$ points on an $m\times n$ grid in $O(\log v)$
  time. It can be built in $O(v\log v)$ time.
\end{lemma}
\begin{proof}
  Using the rank space technique~\cite{GBT84}, we reduce the problem
  of counting on an $m\times n$ grid to the problem of counting on a
  $v\times n$ grid: let $\tau(p.x,p.y)=(\rank(p.x,S_x),p.y)$, where
  the set $S_x$ consists of the $x$-coordinates of all the points in $S$ and
  $\rank(a,S)=|\{\,e\in S\, |\, e\leq a\,\}|$. Then the mapped set is
  $S' = \{ \tau(p.x,p.y), (p.x,p.y) \in S \}$.
  Let $\pred(a,S)=\max\{\,e\in S\,|\,e\leq a\,\}$ and
  $\ssucc(a,S)=\min\{\,e\in S\,|\,e\geq a\,\}$.  Then a query $[a,b]\times
  [0,h]$ on $S$ is equivalent to a query
  $[a',b'] \times [0,h]$ on $S'$, where $a'=\rank(\ssucc(a,S_x), S_x)$ and
  $b' = \rank(\pred(b,S_x), S_x)$.
  Using standard binary search on an array, we can find $a'$ and
  $b'$ in $O(\log v)$ time and $v\log m$ bits of space.
 
We store the points on a variant of the wavelet tree data structure
\cite{GGV03}. Each node of this tree $W$ covers all the points within
a range of $y$-coordinates. The root covers all the nodes, and the two
children of each internal node cover half of the points covered by their
parent. The leaves cover one point. The $y$-coordinate limits of the nodes
are not stored explicitly, to save space. Instead, we store the $x$-coordinate 
of the point holding the maximum $y$-coordinate in the node. With the global
array we can recover the $y$-coordinate in constant time. Each internal node 
$v$ covering $r$ points stores a bitmap $B_v[1..r]$, so that $B_v[i]=0$ iff the
$i$-th point, in $x$-coordinate order, belongs to the left child (otherwise
it belongs to the right child). Those bitmaps are equipped with data
structures answering operation $rank_b(B_v,i)$ in constant time and $r+o(r)$ 
bits of space \cite{Mun96}, where $rank_b(B_v,i)$ is the number of occurrences
of bit $b$ in $B_v[1..i]$. Since $W$ has $O(v)$ nodes and height $O(\log v)$,
its bitmaps require $O(v \log v)$ bits and its pointers and $x$-coordinates
need $O(v\log m)$ bits. The construction time is $O(v\log v)$.

We can easily answer range counting queries $[a',b'] \times [0,h]$ on $W$ 
\cite{MN07}. The procedure starts at the root node of $W$, with the range
$[a',b']$ on its bitmap $B$. This range will become $[a_l,b_l] =
[rank_0(B,a'-1)+1,rank_0(B,b')]$ on the left child of the root, and $[a_r,b_r]
= [rank_1(B,a'-1)+1,rank_1(B,b')]$ on the right child. If the maximal 
$y$-coordinate of the left child is smaller than or equal to $h$, we count the 
number of points $p$ with $p.x\in [a,b]$ stored in the left child, which is
simply $b_l-a_l+1$, and then visit the right child.  Otherwise, the maximal
$y$-coordinate in the left child is larger than $h$, and we just visit the
left child. The time is $O(1)$ per tree level.
\end{proof}

\begin{lemma}\label{lemma:rep1}
  Let $v\leq m \leq n$.  There exists a data structure that uses
  $O(v\log m)$ additional bits and answers three-sided range reporting
  queries for a set $S$ of $v$ points on an $m\times n$ grid in $O(\log v
  + occ)$ time, to report the $occ$ results. It can be built in $O(v\log v)$
time.
\end{lemma}
\begin{proof}
  We can reduce
  the problem of reporting on an $m\times n$ grid to the problem of
  reporting on a $v\times n$ grid, as described above.  The query time
  is increased by an additive $O(\log v)$ factor, and the space usage
  increases by $O(v\log m)$ bits of space.
  Now we sort the $v$ points in $x$-coordinate order, build the sequence
$Y[1..v]$ of 
  their $y$-coordinates, and build a Range Minimum Query (RMQ) data structure 
  on $Y$ \cite{FH11}. This structure requires only $O(v)$ bits of space, does 
not need to access $Y$ after construction (so we do not store $Y$), and answers
  in constant time the query $rmq(c,d) = \mathrm{arg}~\min_{c \le i \le d}
Y[i]$ for any $c,d$.
  It is well known that with such queries one can recursively
  retrieve all the points in the three-sided range in $O(occ)$ time; see, 
  for example, Muthukrishnan~\cite{Mut02}.
The construction time is dominated by the sorting of points.
\end{proof}

\section{An Optimal Time Data Structure}
\label{sec:optimal}

The data structure of Lemma~\ref{lemma:polyn} gives us an $O(m^f+k)$ time 
solution, for any constant $f$, where $m$ is the grid width. In this section 
we use it to obtain $O(h/\log n+k)$ time.
The idea is to partition the space into vertical stripes, for different
stripe widths, and index each stripe with Lemma~\ref{lemma:polyn}. Then the
query is run on the partition of width $m$ so that the $O(m^f)$ time 
complexity is dominated by $O(h/\log n+k)$.
The many partitions take total linear space because the size per point in
Lemma~\ref{lemma:polyn} is $O(\log m)$, and our widths decrease doubly
exponentially. As a query may span several stripes, a structure similar to 
the one used in the classical RMQ solution \cite{BF00} is used.
This gives linear space for stripes of width up to 
$\Omega(\log^2 n)$. 
Smaller ones are solved with universal tables.

In addition to the global array storing
$p.y$ for each $p.x$, we use another array storing the weight corresponding to
each $p.x$. As there are overall $O(n)$ different weights, those can be mapped 
to the interval $[1,O(n)]$ and still solve correctly any top-$k$ reporting 
problem. Thus the new global array also requires $O(n)$ words of space.

\subsection{Structure} \label{sec:struc}

Let $g_j=1/2^j$ for $j=0,1,\ldots, r$. We choose $r$ so that
$n^{g_r}=O(1)$, thus $r=O(\log\log n)$.  
The $x$-axis is split into intervals of size
$\Delta_j=n^{g_j}\log^2 n$ and
$j=1,\ldots, r$.  For convenience, we also define $\Delta_0=n$ and
$\Delta'_j=\Delta_j/\log^2 n=n^{g_j}$.
For every $1\leq j < r$ and for every interval
$I_{j,t}=[(t-1)\Delta_j,t\Delta_j-1]$, we store all the points $p$ with 
$p.x\in I_{j,t}$ in a data structure
$E_{j,t}$ implemented as described in Lemma~\ref{lemma:polyn}.
$E_{j,t}$ supports three-sided top-$k$ queries in $O((\Delta_j)^{f}+k)$ time
for any constant $0<f<1/4$. We also construct a data structure $E_0$ 
that contains all the points of $S$ and supports three-sided top-$k$
queries in $O(n^{1/4}+k)$ time. To simplify the description, we 
also assume that $I_{-1}=I_0=[0,n-1]$ and $E_{-1}=E_0$.

The data structures $E_{j,t}$ for a fixed $j$ contain $O(n)$ points overall,
hence by Lemma~\ref{lemma:polyn} all $E_{j,t}$ use $O(n\log\Delta_j) =
O(n\log(n^{g_j}\log^2 n))=(1/2^j)O(n\log n) + O(n\log\log n)$ 
additional bits of space.  Thus all $E_{j,t}$ for $0\leq j < r$ use
$\sum_{j=0}^{r-1} [(1/2^j)O(n\log n) + O(n\log\log n)]=O(n\log n)$ bits, or
$O(n)$ words. They also require $O(n\log n)$ total construction time.
Since $f < 1/4$, a data structure $E_{j,t}$ supports top-$k$ queries in 
time $O((\Delta_j)^{f}+k) = O((n^{g_j}\log^2 n)^{f}+k)=
O(n^{g_{j+2}}/\log n+k)= O(\Delta'_{j+2}/\log n+k)$ time.
For each of the smallest intervals $I_{r,t}$ we store data structures $\oE_{r,t}$
that use $o(\log^2 n)$ words of space (adding up to $o(n)$) and support
three-sided top-$k$ queries in $O(h/\log n+k)$ time. This structure will be
described in Section~\ref{sec:smallgrids}.

Note that our choice of writing $(n^{g_j}\log^2 n)^{f} = O(n^{g_{j+2}}/\log
n)$ was arbitrary, because $f$ is strictly less than $1/4$. 
We could have written $(n^{g_j}\log^2 n)^{f} =
O(n^{g_{j+2}}/\log^c n)$ for any constant $c$, and this would yield
$O(h/\log^c n+k)$ query time. We have chosen to favor simplicity in the
exposition, but will
return to this point at the end of Section~\ref{sec:smallgrids}.

\subsection{Queries}

We can carry out the query using a range of intervals $I_{j,t}$ of any width
$\Delta_j$. The key idea is to use a $j$ value according to the height of the
three-sided query, so that the search time in $I_{j,t}$ gives the desired
$O(h/\log n)$ time. More precisely,
assume we want to report $k$ points with highest weights in the
range $[a,b]\times [0,h]$.  First, we find the index $j$ such that
$\Delta'_{j+1} > \max(h,k) \geq \Delta'_{j+2}$.  The index $j$ can be
found in $O(\log\log(h+k))$ time by linear search\footnote{This is
$O(h/\log n + k)$ for sure if $\max(h,k) = \Omega(\log n\log\log n)$; 
otherwise a small table can be used to perform the search in constant time.}.  
If $[a,b]$ is
contained in some interval $I_{j,t}$, then we can answer a query in
$O(\Delta'_{j+2}/\log n+k) = O(h/\log n+k)$ time using $E_{j,t}$. 
If $[a,b]$ is contained 
in two adjacent intervals $I_{j,t}$ and $I_{j,t+1}$, we generate the lists of
top-$k$ points in $([a,b]\cap I_{j,t})\times [0,h]$ and $([a,b]\cap
I_{j,t+1})\times [0,h]$ in $O(h/\log n+k)$ time, and merge them in $O(k)$ time.  
To deal with
the case when $[a,b]$ spans one or more intervals $I_{j,t}$, we store
pre-computed solutions for some intervals.

For $1\leq j\leq r$, we consider the endpoints of intervals $I_{j,t}$.
Let $\mytop_j(a,b,c,k)$ denote the list of top-$k$ points in the range 
$[a\cdot\Delta_j,b\cdot\Delta_j-1]\times [0,c]$ in descending weight order.  
We store the values of 
$\mytop_j(t,t+2^v,c,\Delta'_{j+1})$ for any 
$t\in [0,n/\Delta_j]$, any $0\leq v\leq \log_2 (n/\Delta_j)$, and any 
$0\leq c\leq \Delta'_{j+1}$.  All the lists
$\mytop_j(\cdot,\cdot,\cdot,\cdot)$ use space
$O((n/\Delta_j)(\Delta'_{j+1})^2\log n)=O(n/\log n)$ words. 
Hence the total word space usage of all lists 
$\mytop_j(\cdot,\cdot,\cdot,\cdot)$ 
for $2\leq j\leq r$ is $O(n\log\log n/\log n) = o(n)$.
It can also be built in $o(n)$ time using dynamic programming.

Assume that $[a,b]$ spans intervals $I_{j,t_1+1},\ldots, I_{j,t_2-1}$;
$[a,b]$ also intersects with intervals $I_{j,t_1}$ and $I_{j,t_2}$.  Let
$a' \Delta_j$ and $b' \Delta_j$ denote the left endpoints of $I_{j,t_1+1}$ and 
$I_{j,t_2}$, respectively. The list $L_m$ of top-$k$ 
points in $[a' \Delta_j,b' \Delta_j-1]\times [0,h]$ can be generated as follows.
Intervals $[a'\Delta_j,(a'+2^v)\Delta_j-1]$ and
$[(b'-2^v)\Delta_j,b'\Delta_j-1]$ 
for $v=\ceil{\log_2((b'-a')/\Delta_j)}$ cover
$[a'\Delta_j,b'\Delta_j-1]$. Let $L'_m$ and $L''_m$ denote the lists of the 
first $k$ points in $\mytop_j(a',a'+2^v,h,\Delta'_{j+1})$ and
$\mytop_j(b'-2^v,b',h,\Delta'_{j+1})$ (we have $k$ results 
because $k < \Delta'_{j+1}$; similarly we have the results for $c=h$ because
$h < \Delta'_{j+1}$). We merge both lists 
(possibly removing duplicates) according to the
weights of the points, and store in $L_m$  the set of the first $k$ points 
from the merged list.  Let $L_{t_1}$ and $L_{t_2}$ denote the sets of top-$k$ points in
$[a,a'\Delta_j-1]\times [0,h]$ and $[b'\Delta_j,b]\times [0,h]$. We can 
obtain $L_{t_1}$ and $L_{t_2}$ in $O(h/\log n+k)$ time using data structures $E_{j,t_1}$ and
$E_{j,t_2}$ as explained above.  Finally, we can merge $L_m$, $L_{t_1}$, and
$L_{t_2}$ in $O(k)$ time; the first $k$ points in the resulting list are
the top-$k$ points in $[a,b]\times [0,h]$. 

\subsection{A Data Structure for an $O(\log^2 n)\times n$ Grid.}
\label{sec:smallgrids}

The data structures $\oE_{r,t}$ for an interval $I_{r,t}$ use the same
general approach as the data structures $E_{j,t}$, at a smaller scale.
Note that these structures will be consulted only when $\max(h,k) < C =
\Delta_{r+1}'=O(1)$.
Each interval $I_{r,t}$ is subdivided into $\log^{7/4}n$
intervals $\tI_1,\tI_2,\ldots$ of width $\log^{1/4} n$. 
Let $\tS$ denote the set that contains the
endpoints of $\tI_1,\tI_2,\ldots$. For every $x\in \tS$, each $1\leq
v\leq 2\log\log n$ and each $h\leq C$, we store the lists
$\mytop(x,x+2^v,h,C)$.
All such lists use $O((n/\log^2 n)C^2\log^{7/4}n\log\log n) = o(n)$ space in total.

A query on $I_{r,t}$ is processed as follows. Suppose that 
$[a,b]$ intersects with intervals $\tI_{g_1},\ldots,\tI_{g_2}$ for some $g_1\leq g_2$. 
We find the top-$k$ points from
$(\tI_{g_1+1}\cup\ldots\cup\tI_{g_2-1})\times [0,h]$ in $O(k)$ time using lists 
$\mytop(\cdot,\cdot,\cdot,\cdot)$, as before. We also find top-$k$ points from
$(\tI_{g_1}\cap [a,b])\times [0,h]$ and $(\tI_{g_2}\cap [a,b])\times [0,h]$ 
in $O(k)$ time using data structures for $\tI_{g_1}$ and $\tI_{g_2}$,
respectively, to be described next. 
We thus obtain three lists of points sorted by their 
weights, and merge them in $O(k)$ time as before. 
 
Finally, we describe how to answer queries in $O(k)$ time in the grids
$\tI_g$, of width $\log^{1/4} n$. We replace the 
$y$-coordinates of points by their ranks; likewise, the weights of
points are also replaced by their ranks. The resulting sequence $X_g$
contains all mapped points in $\tI_g$ and consists of 
$O(\log^{1/4} n\log\log n)$ bits, so all the descriptions of all sequences
$X_g$ require $O(n\log\log n)$ bits, or $o(n)$ words. 
There are $O(\log^{1/2} n)$ queries
that can be asked (considering all the sensible values of $[a,b]$, $h$ and $k$),
and the answers require $O(k\log(\log^{1/4} n)) =
O(\log\log n)$ bits.
Thus we can store a universal look-up table of size 
$2^{O(\log^{3/4} n \log\log n)} O(\log\log n) = o(n)$ words
common to all subintervals $\tI_g$. This table contains pre-computed
answers for all possible queries and all possible sequences $X_g$. Hence,
we can answer a top-$k$ query on $X_g$ in $O(k)$ time. 

A query on $\tI_g$ can be transformed into a query on $X_g$ by reduction to
rank space in the $y$ coordinates. 
Consider a query range $Q=[a,b]\times[0,h]$ on $\tI_g$. 
We can find the rank $h'$ of $h$ among the $y$-coordinates of points 
from  $\tI_g$ in $O(h)=O(1)$ time by linear search (remember that we store 
only the reordering of the local $x$-coordinates, and the actual $y$-coordinates
are found in the global array). 
Then, we can identify the top-$k$ points in $X_g\cap Q'$,
where $Q'=[a,b]\times [0,h']$, using 
the look-up table and report those points in $O(k)$ time. 

\medskip

Thus our data structure uses $O(n)$ words of space and answers 
queries in $O(h/\log n+k)$ time. It can be built in $O(n\log n)$ time.
As mentioned at the end of Section~\ref{sec:struc}, we can obtain any query 
time of the form $O(h/\log^c n+k)$, for any constant $c$. This completes the 
proof of Theorem~\ref{theor:weight}, which is given in this general form.

\subsection{Online Queries}
\label{sec:online}

An interesting extension of the above result is that we can deliver the
top-$k$ documents in online fashion. That is, after the $O(p/\log_\sigma n)$ 
time initialization,
we can deliver the highest weighted result, then the next highest one, and so 
on. It is possible to interrupt the process at any point and spend overall time
$O(p/\log_\sigma n+k)$ after having delivered $k$ results. That is, we obtain the same
result without the need of knowing $k$ in advance.
This is achieved via an online version of Theorem~\ref{theor:weight}, and is
based on the idea used, for example, in \cite{BFGLO09,HSV09}.  We describe it
for completeness.

Consider an arbitrary  data structure that answers top-$k$ queries 
in $O(f(n)+kg(n))$ time in the case when $k$ 
must be specified in advance. Let $k_1=\ceil{f(n)/g(n)}$, $k_i=2k_{i-1}$, 
and $s_i=\sum_{j=1}^{i-1}k_j$ for $i\geq 2$. 
Let $S$ be the set of points stored in the data structure and 
suppose that we must report top points from the range $Q$ in the 
online mode.  
At the beginning, we identify top-$k_1$ points in $O(f(n)+g(n))$ time
and store them in a list $L_1$.
Reporting is divided into stages. During the 
$i$-th stage, we report points from a list $L_i$. 
$L_i$ contains $\min(k_i, |Q\cap S|-s_i)$ top points that were not 
reported during the previous stages.  Simultaneously we compute the list 
$L_{i+1}$ that contains $\min(2k_i+s_i, |Q\cap S|)<4k_i$ top points. 
We  identify at most $2k_i+s_i$ top points in 
$O(f(n)+4k_i\cdot g(n))=O(k_i\cdot g(n))$ time. 
We also  remove the first $s_i$  points from $L_{i+1}$ in $O(k_i)$ time.
The resulting list $L_{i+1}$ contains $2k_i=k_{i+1}$  points that  
must be reported during the next $(i+1)$-th stage. 
The task of creating  and cutting the  list $L_{i+1}$ is executed in such 
a way that we spend $O(1)$ time when each point of $L_i$ is reported. 
Thus when all the points from $L_i$ are output, the list $L_{i+1}$ that 
contains the next $k_{i+1}$ top points is ready and we can proceed with 
the $(i+1)$-th stage.   

The reporting procedure described above outputs the first $k$ most highly 
weighted points in $O(f(n)+kg(n))$ time. 
It can be interrupted at any time.

\section{A Dynamic Structure} 
\label{sec:dynamic}

We consider now a scenario where insertions and deletions of whole documents
are interspersed with top-$k$ queries. This has two important components: how
to maintain the suffix tree and how to maintain the grid. Then we consider how
queries are performed. We first consider just the measure $\tf$ and then
generalize to others at the end.

\subsection{Dynamic Suffix Trees}
\label{sec:dynst}

We use a dynamic suffix tree maintenance
algorithm where leaves and unary nodes can be inserted and deleted, and lowest
common ancestors can be computed, all in constant worst-case time \cite{CH05}.
The updates on leaves and unary nodes are the operations we need to insert and 
delete all the suffixes of a document in the suffix tree, in time proportional 
to the length of the document inserted or deleted, whereas the lowest common
ancestor queries are necessary to compute the new $\ptr(\cdot,\cdot)$ pointers
to insert. In any node we will maintain the up to $\sigma$ children using a
linear-space dynamic predecessor data structure that supports queries and 
updates in worst-case time $o((\log\log\sigma)^2)$ \cite{AT07}. 

Upon insertion of a new document $d$ of length $|d|$, we follow McCreight's
procedure to insert a new string in a generalized suffix tree \cite{McC76}. 
First we search for the new document string in the suffix tree, add the 
corresponding leaf and, if necessary, insert its parent splitting an edge. 
Then we compute the suffix link of the last node that
belonged to the path found, follow the suffix link, descend again as much as 
possible, create a new leaf and possibly an internal node, create suffix links 
from the last to the current created nodes, and continue until exhausting the
suffix. This requires $O(|d|)$ suffix tree operations, and total time
$O(|d|(\log\log\sigma)^2)$.

The deletion of a document $d$ is symmetric to insertion.
We find its corresponding string, delete its leaf and possibly its parent if 
it becomes unary, follow the suffix link, and repeat the process until removing
all the leaves and possibly their parents. 
This takes $O(|d|(\log\log\sigma)^2)$ time.

In order to accelerate searches we will use a technique analogous to the one 
used with the static suffix tree. We define $\ell=\log_\sigma n$, and the 
{\em level} of a node $v$ as $lev(v) = \lfloor |l(root,v)| / \ell \rfloor$. 
Note that the level of a node depends on its string depth, and thus it does 
not change upon updates. Each suffix tree node $v$ with parent $u$ such that 
$lev(v) > lev(u)$ will maintain a predecessor data structure called an
{\em accelerator}, storing all its highest descendant nodes $z$ such that 
$lev(z) > lev(v)$. The key used for the predecessor data structure are the 
$\ell$ characters ($\lg n$ bits) formed by $l(root,z)[lev(v)\cdot\ell+1,
(lev(v)+1)\cdot\ell]$. Note these keys do not depend precisely on $v$ being
the node holding the accelerator; any other ancestor of $z$ of
the same level of $v$ yields the same key. Note also that the nodes $z$ stored 
in the accelerator of $v$ are owners of subsequent accelerators. 

The predecessor structures hold $O(n)$ nodes, and thus they require
$o((\log\log n)^2)$ time and linear space \cite{AT07}. The total extra space
is linear because each suffix tree node belongs to at most one predecessor
structure.

Upon searches, we start at the root and use the accelerators of 
successive nodes, using consecutive chunks of $\ell$ symbols in $P$. In some
cases we may arrive at nodes whose string distance to the previously visited 
node is more than $\ell$; in those cases we check the missing symbols directly
in the text, also in chunks of $\ell$ characters. When, finally, the 
accelerator does not give a node matching the next $\ell$ symbols of $P$, we
switch to the character-based search. Thus the total search time is
$O(p(\log\log n)^2/\log_\sigma n + (\log_\sigma n)(\log\log\sigma)^2) =
O(p(\log\log n)^2/\log_\sigma n + \log n)$.

Those accelerators must be updated upon insertions and deletions
of suffix tree nodes. Note that we always know $l(root,v)$ when we insert or
delete a node $v$. Upon insertion of a leaf $v$ as a child of a node $u$, it 
may turn out that the leaf must be inserted into an accelerator
(because $lev(v) > lev(u)$). We can simply find the nearest ancestor holding
an accelerator via at most $\ell$ parent operations from $v$. We must
also initialize an empty accelerator for $v$. Symmetrically, when
a leaf $v$ is removed, we may have to remove it from its ancestor's accelerator.
When an edge from $u$ to $v$ is split with a new node $z$, it may
be that $lev(u) < lev(z) = lev(v)$. In this case, $z$ takes the role of $v$,
``stealing'' the accelerator from $v$ (which needs no change, as
explained). We must also replace $v$ by $z$ in the accelerator stored at
the proper ancestor of $z$. Another case that requires care is when
$lev(u) < lev(z) < lev(v)$. In this case $v$ is replaced by $z$ in the 
proper ancestor of $z$, but $v$ retains its accelerator and $z$
creates a new accelerator holding only $v$. Other cases require no
action. Upon deletions, the obvious reverse actions are necessary.
The total update time can be bounded by $O(|d|\log n)$ for both insertions and
deletions.

\subsection{Relating the Suffix Tree and the Grid}
\label{sec:relating}

Since grid columns will appear and disappear upon document insertions and
deletions, we will not associate integers to columns, but just abstract labels.
The mapping between the suffix tree and the grid columns will be carried out
via a dynamic technique to maintain order in a list $X$ of such abstract labels 
\cite{DS87,BCDFCZ02}. The data structure supports the operations of creating a 
new label $y$ as the immediate successor of a given label $x \in X$, deleting 
a label $y \in X$, 
and determining which of two labels comes first in $X$, all in constant time. 
In addition, each suffix tree node $v$ will hold a (classical) doubly-linked 
list $\mathit{list}(v)$ storing consecutive labels of $X$, 
each label corresponding to a grid column where this node induces points, 
and will maintain pointers to the first and the last node in $\mathit{list}(v)$.
Finally, $v$ will maintain special labels
$\mathit{first}(v), \mathit{last}(v)\in X$ that do not represent any column, 
but are the predecessor (resp.\ successor) in $X$ of the first (resp.\ last) 
label in its subtree.

As we insert a new leaf $v$ as the child of $u$, we must create a predecessor of
$\mathit{last}(v')$ to assign to $\mathit{last}(v)$, where $v'$ is the next
sibling of $v$, and then create a predecessor of $\mathit{last}(v)$ to assign to
$\mathit{first}(v)$. If $v$ is the last child of $u$, then we create a 
predecessor of $\mathit{last}(u)$ to assign to $\mathit{last}(v)$.
When, instead,
we create a new node $v$ that splits an edge from $u$ to $v'$, then 
$\mathit{first}(v)$ will be a new predecessor of $\mathit{first}(v')$ and
$\mathit{last}(v)$ will be a new successor of $\mathit{last}(v')$. When a node 
$v$ is removed, its labels $\mathit{first}(v), \mathit{last}(v)$ are also 
removed from $X$.

As we insert a new document $d$, we must associate new grid columns to the new 
and existing suffix tree nodes traversed. Each newly created pointer 
$\ptr(v,d)$ will require creating a new label $t(v,d) \in X$ as the successor 
of the last one currently held in $\mathit{list}(v)$ (it will
also be stored at the end of $\mathit{list}(v)$). When $\mathit{list}(v)$ is
empty, the new label $t(v,d)$ must be created as the successor of 
$\mathit{first}(v)$.

As we insert new leaves in the suffix tree, we collect them in an array
$L[1,|d|]$. We also create the first label $t(v,d)$ of such leaves $v$. Now we 
sort $L$ by the labels $t(v,d)$, and as a result the new leaves
become sorted by their suffix tree preorder. All the internal suffix tree nodes
that must be labeled with $d$ are obtained as $u=lca(v,v')$ for consecutive
leaves $v=L[i]$ and $v'=L[i+1]$. We create pointers $\ptr(v,d)$ and 
$\ptr(v',d)$ towards node $u$, associated to the labels 
$t(v,d)$ and $t(v',d)$, respectively, and with weights $w(v,d)=w(v',d)=1$.
For each new internal suffix tree node $u=lca(v,v')$ obtained, we create a new
label $t(u,d)$ for the new grid column that $u$ will originate, and associate 
weight $w(u,d)=2$ to it (at the end, $w(u,d)$ will be the number of leaves 
labeled $d$ in the subtree of $u$). Each time $u$ is obtained again (which we
know because the last element of $\mathit{list}(u)$ is already $t(u,d)$), we 
increase $w(u,d)$ by $1$. 

Now we have to propagate weights and pointers from internal nodes labeled with 
$d$ to their nearest ancestors labeled with $d$ (i.e., the nodes that would
be their parent in the suffix tree of document $d$). 
For this sake, the internal nodes $u=lca(v,v')$ 
obtained are collected in a new array $I$, of size up to $|d|-1$, and $I$ is 
sorted by the labels $t(u,d)$, so that the nodes become sorted by preorder.
We traverse $I$ left to right, simulating a recursive preorder traversal of
the suffix tree of document $d$, although the nodes are in the generalized
suffix tree. Let $u=I[i]$ and $v=I[j]$, initially for $i=1$ and $j=2$. If
$lca(u,v)=u$, then $u$ is the parent of $v$ in the suffix tree of $d$.
Thus we recursively traverse the subtree that
starts in $v=I[j]$, which finishes at a node $v'=I[j']$ that is not anymore
a descendant of $v$. Now we check whether $lca(u,v')=u$ (i.e., $v'$ is the
second child of $u$ in the suffix tree of $d$), and so on.
At some point, it will hold that $I[j']$ does not
descend from $u$, and we have finished the traversal of the subtree of $u$.
Along this recursive traversal we will identify the nearest ancestor $u$
labeled $d$ of each node $v$ labeled $d$. For each such pair, and after having
processed $v$ and computed $w(v,d)$, we increase $w(u,d) = w(u,d)+w(v,d)$
and generate the pointer $\ptr(v,d)$ pointing to $u$, associated to label
$t(v,d)$ and weight $w(v,d)$.

All the labels created when inserting a document $d$ are additionally chained
in a (classical) list $\mathit{list}(d)$, to facilitate deletion of
the document. The overall time of this step is $O(|d|\log |d|)$, dominated by 
the sorting via comparisons of labels in $X$.

Finally, we will create new columns and points in the grid associated to all 
the pointers $\ptr(v,d)=u$ created. The label $t(v,d)$ will be an identifier 
for the $x$-coordinate of the point (we remark that these are not integers, but
just labels that can be compared). The $y$-coordinate will be the string depth
of the target node, $|l(root,u)|$. This value is stored in the suffix tree node
when the node is created and, unlike the tree depth, does not change upon 
suffix tree updates. The document associated to the new point is the new one, 
$d$, and the weight is the value $w(v,d)$ associated to the source node of the 
pointer. 

Handling the deletion of a document $d$ is simple.
After deleting all the corresponding suffix tree nodes, we follow the chain of 
labels $t(v,d)$ in $\mathit{list}(d)$, delete them from $X$ and remove 
their nodes from the doubly-linked list $\mathit{list}(v)$. This takes 
$O(|d|)$ additional 
time. We also remove the columns in the grid corresponding to the labels 
deleted, and the associated points. Both insertion and deletion times are 
superseded by those of Section~\ref{sec:dynst}.

\subsection{Slim Grids} \label{sec:slimgrid}

To achieve faster searches, the grid will be divided into horizontal slices of
small height $r$. For every slice, we maintain a structure that 
reports $k$ most highly weighted points from a horizontal range of labels 
$[a,b)$ intersected with a vertical range of integers $[0,y)$. We describe
here how those slices are updated and queried for a sublogarithmic value of
$r$, and in Section~\ref{sec:multires} we extend the solution to grids of 
polylogarithmic height.

Each slice is represented with a B-tree ordered by the labels
(i.e., $x$-coordinates) of the points, of arity $\rho$ to $2\rho-1$, for some 
sublogarithmic $\rho$ to be defined later (as usual, the root can have arity
as low as 2). Thus the B-tree has height $O(\log_\rho n)$.
 At each internal node $u$ with $a(u)$ children $v_1,\ldots,
v_{a(u)}$, we will store $a(u)$ arrays $W_{v_1}[0..r-1],\ldots,
W_{v_{a(u)}}[0..r-1]$. 
In these arrays $W_v$, $W_v[y]$ is the point $p$ with maximum
weight among all points (1) whose $x$-coordinates belong to the 
subtree of $v$, (2) with $y$-coordinate equal to $y$, and (3) 
not stored in $W_u[y]$ for ancestors $u$ of $v$ (some $W_v[y]$ cells can be
empty, if no point with $y$-coordinate $y$ exists below $v$). 
We will also store a structure $W_{root}$ for the root node $root$.
Thus $W_{root}[y]$ contains the point $p_r$ of maximum weight among all 
points with $y$-coordinate $y$; for a child $v$ of $root$, $W_{v}[y]$ contains 
the point $p_v$ of maximum weight among all points $p\not=p_r$ in the subtree
of $v$ with $y$-coordinate $y$. In general, 
all points already stored in ancestors are excluded from consideration.
We store $W_v[y] = (x,w,d)$, where $x$ is the $x$-coordinate, $w$ is the weight, and $d$
is the document of the point. 
Each point is also stored in the corresponding leaf
node. Those points in $W_v$ are not used to separate the $x$-coordinates of 
the points in the tree. Instead, new labels $x(v_1) \ldots x(v_{a(u)-1})\in X$ 
will be created and stored at node $u$, to split the points between 
its $a(u)$ consecutive children. That is, the $x$-coordinate of any
point stored below $v_i$ will be between $x(v_{i-1})$ and $x(v_i)$. 
The size of the list $X$ stays $O(n)$.

The leaves of the B-tree will store $r$ to $2r-1$ points. Leaves store the 
actual points, even if they are also mentioned in some previous $W_v$ structure.
The points in leaves $l$ are arranged in an array $W'_l$, which is similar to
the arrays $W_v$ and list the points in increasing $y$-coordinate order, except
that $W'_l$ has no empty cells and some $y$-coordinates can be repeated in the
points. Therefore the $W'_l$ cells store the full point data,
$W'_l[j] = (x,y,w,d)$.

To each internal node $u$ with children $v_1,\ldots,v_{a(u)}$ we will also 
associate structures $\Yx_u[0..a(u)r-1]$ and $\Yw_u[0..a(u)r-1]$, where the 
child numbers and the $y$-coordinates of the (up to) $r$ points of the $a(u)$ 
arrays $W_{v_i}$ are sorted by their $x$-coordinate 
label (in $\Yx_u$) and by their weight (in $\Yw_u$). That is, in $\Yx_u$ and 
$\Yw_u$ we store the pair $(i,y)$ for each entry $W_{v_i}[y]$, ordered by 
$W_{v_i}[y].x$ (in $\Yx_u$) or by $W_{v_i}[y].w$ (in $\Yw_u$). 
Each value stored in $\Yx_u$ and $\Yw_u$ requires $\lg (2\rho r)$ bits,
thus all the values in these two structures add up to at most
$4\rho r\lg (2\rho r)$ bits. 
Each time we modify a value in a $W_{v_i}$ array, we rebuild from scratch the 
$\Yx_u$ and $\Yw_u$ structures of the parent $u$ of $v_i$.

We will also maintain structures $\Yx_l'$ and $\Yw_l'$ on the (up to) $2r-1$
points of leaves $l$, analogous to the $\Yx_u$ and $\Yw_u$ structures of 
internal nodes. Instead of the pairs $(i,y)$, structures $\Yx_l'$ and $\Yw_l'$ 
will just store positions $j$ of the array $W_l'$ (those positions would 
coincide with $y$-coordinates in internal nodes). Leaves will also store an
array $Y_l[0..r-1]$ where $Y_l[y]=j$ if $j$ is the last position where
$W'_l[j] < y$. Finally, leaves will store 
bitmaps $Q_l$ marking in $Q_l[j]$ whether the point in $W_l'[j]$ also appears 
in the $W_v$ array of an ancestor $v$ of $l$.

We will end up choosing $\rho r=o(\log n)$, and thus 
universal tables of $2^{4\rho r\lg (2\rho r)}\cdot O(\mathrm{polylog}~(\rho r))
= o(n)$ bits will be used to query and update the
arrays $\Yx_u$ and $\Yw_u$, in constant time. Similarly, leaves will use even
smaller universal tables of $2^{4r\lg (2r)}\cdot O(\mathrm{polylog}~(r))=o(n)$
bits.

The whole data structure requires linear space, because the leaves contain
$\Theta(r)$ points. The $W_v$ arrays of internal nodes spend $\Theta(r)$ 
words and can 
be almost empty (if all the descendants have the same $y$-coordinate,
say), but there are only $O(n/r)$ internal nodes. If the whole grid contains 
less than $r$ points, we just store the space for them in a leaf.

\subsubsection{Insertions}

Consider the insertion of a new point $(x,y,w,d)$, 
with label $x \in X$, $y$-coordinate $y \in [0,r)$, weight $w$ and document $d$.
While following the normal insertion procedure on the B-tree 
(where we compare the labels $x(v_i)$ of the nodes with $x$ to decide 
the insertion path), we look for the highest node $v$ with $W_v[y].w < w$ or 
with $W_v[y]$ empty. For the first (i.e., highest) such $v$ we find, 
we set $W_v[y] \leftarrow (x,w,d)$, and then we continue the classical insertion
procedure (not looking at $W_v[y]$ entries anymore) until adding the
point $(x,y,w,d)$ in a leaf $l$. In the leaf we mark in
the corresponding $Q_l$ entry whether we had updated an entry $W_v[y]$ in
some ancestor $v$.

If we updated some $W_v[y]$, and it already had a previous value 
$W_v[y]=(x',w',d')$, we perform a process we call {\em reinsertion} of 
$(x',w',d')$.
We restart the process of inserting the point $(x',y,w',d')$ from node $v$
(note that this point already exists in a leaf; reinsertion will not alter the
structure of the tree, but just rewrite some $W$ and $Q$ values). 
In the reinsertion path, if 
we arrive at a node $v'$ where $W_{v'}[y].w < w$, we set
$W_{v'}[y] \leftarrow (x',w',d')$. If there was a previous value 
$W_{v'}[y]=(x'',w'',d'')$, we continue the reinsertion process for point
$(x'',w'',d'')$ from node $v'$, and so on until either we find an empty space 
in some $W_u[y]$ or we reach the leaf $l$ where the point being reinserted is 
actually stored. In this latter case, we clear the corresponding bit in $Q_l$, 
indicating that this point is not stored anymore in an ancestor structure.

Thus we traverse two paths, one for inserting the point, and another for
reinserting the point(s) possibly displaced from some $W_v[y]$ structure.
This part of the operation requires, in the worst case, $O(\log_\rho n)$ updates
to the structures $\Yx_u$ and $\Yw_u$ of the parents $u$ of nodes $v$ where 
$W_v$ is modified, plus an insertion in a leaf.

\paragraph{Rebuilding structures $\Yx_u$ and $\Yw_u$.}
Upon an assignment $W_{v_i}[y] \leftarrow (x,w,d)$, we must rebuild the 
structures $\Yx_u$ and $\Yw_u$ of the parent $u$ of $v_i$. 
We binary search $\Yx_u$ for $x$, and binary search $\Yw_u$ for $w$,
both in $O(\log (\rho r))$ time. 
In these binary searches we obtain the actual label and weight of each element 
of $\Yx_u$ and $\Yw_u$, respectively, using its $(i,y)$ pair, as 
$W_{v_i}[y].x$ and $W_{v_i}[y].w$.
These binary searches give the insertion positions $0 \le e < 2\rho r$ and 
$0 \le g < 2\rho r$, respectively, of the pair $(i,y)$ in $\Yx_u$ and $\Yw_u$.
Note that the new contents of $\Yx_u$ and $\Yw_u$ depend only on their current 
contents, on the values $e$ and $g$, and on the incoming pair $(i,y)$ (the
existing occurrence of $(i,y)$, if any, must be removed). Thus, the new content
of $\Yx_u$ and $\Yw_u$ for each $(e,g,i,y)$ can be precomputed in a universal 
table of $o(n)$ bits, as explained, so that they are updated in constant time.
Therefore the time to update the structures is $O(\log (\rho r))$, and the
cost of a full reinsertion process is $O(\log_\rho n \log (\rho r))$.

Removing the value of a cell $W_{v_i}[y]$ is analogous. We can insert all the
cells of a whole new $W_v$ array, or remove all the cells of a whole $W_v$ 
array, one by one in time $O(r\log (\rho r))$. When a whole array is inserted
or removed, we have to rename all the labels $i$ in the cells $(i,y)$, but 
those updates can also be precomputed in universal tables of sublinear size.

\paragraph{Insertion in leaves.}
In leaves $l$, we must actually insert the point, possibly displacing all the
entries in $W_l'$ and recalculating $Y_l$, $\Yx'_l$, $\Yw'_l$ and $Q_l$, all
in $O(r)$ time. When a leaf overflows to $2r$ points, we
must split it into two leaves $l'$ and $l''$ of $r$ points each. We first
remove the array $W_l$ from the parent $u$ of $l$, clearing the corresponding
bits in $Q_l$. Now we distribute the points of $W_l'$ into the new arrays 
$W_{l'}'$ and $W_{l''}'$, and make $l'$ and $l''$ children of $u$, replacing 
the old $l$. We create a new label $x(l') \in X$ right after the largest 
$x$-coordinate in $l'$, and add it to $u$ separating $l'$ and $l''$.

Now we build new arrays $W_{l'}$ and $W_{l''}$. Those arrays, as well as the 
$Y$, $\Yx'$, $\Yw'$ and $Q$ structures of $l'$ and $l''$, are built in $O(r)$ 
time 
from $W_l'$, $\Yx_l'$, $\Yw_l'$ and $Q_l$. We mark in $Q_{l'}$ and $Q_{l''}$ the
points that have been included in $W_{l'}$ and $W_{l''}$ (we cannot choose
any point for $W_{l'}$ and $W_{l''}$ that is already marked in $Q_l$).
Finally, we insert $W_{l'}$ and $W_{l''}$ in $u$.

The overall time is $O(r)$, but this is dominated by the $O(r\log(\rho r))$ 
time needed to update the $\Yx_u$ and $\Yw_u$ arrays upon the $O(r)$ changes
induced by substituting $W_l$ by $W_{l'}$ and $W_{l''}$.

\paragraph{Overflows in internal nodes.}
Further, the insertion in the parent $u$ can trigger an overflow in this
internal node, if its arity reaches $2\rho$. We must split $u$, with children
$v_1,\ldots,v_{2\rho}$, into two nodes, $u'$ with children $v_1,\ldots,v_\rho$
and $u''$ with children $v_{\rho+1},\ldots,v_{2\rho}$.
The process is analogous to the case of leaves,
but slightly more complicated. We create a new $x$-coordinate $x(u') \in X$ 
following $x(v_\rho)$, to separate the points of $u'$ and $u''$. 
We create the two nodes with their corresponding arrays $W_{v_i}$, and build 
the tables $\Yx$ and $\Yw$ of $u'$ and $u''$, in $O(\rho r)$ time from 
$\Yx_u$ and $\Yw_u$. 

Now we must create new arrays $W_{u'}$ and $W_{u''}$ to replace $W_u$ in the 
parent of $u$. First, we move the points in $W_u[y]$ into $W_{u'}[y]$ and 
$W_{u''}[y]$, according to their $x$-coordinate. Now we can get rid of $W_u$, 
but we still have several empty cells in $W_{u'}[y]$ and $W_{u''}[y]$. Those 
are filled with a process we call {\em uninsertion}: To fill some cell 
$W_{u'}[y]$ (analogously for $u''$), we take the maximum weight in cells 
$W_{v_1}[y],\ldots,W_{v_\rho}[y]$. The maximum $W_{v_i}[y].w$ 
is found in constant time using a universal table on $\Yw_{u'}$.
Then we copy $W_{u'}[y] \leftarrow W_{v_i}[y]$, and continue the 
uninsertion process for $W_{v_i}[y]$. 
When we finally arrive at uninserting a point
from a leaf $l$, all we have to do is to mark the corresponding entry in $Q_l$.
Note that uninsertion does not alter the structure of the tree; it just 
rewrites some $W$ and $Q$ values.
The cost of one uninsertion is
$O(\log_\rho n \log(\rho r))$, to rebuild the affected structures $\Yx$ and
$\Yw$. Thus the $O(r)$ uninsertions in $u'$ and $u''$ add up to 
$O(r\log_\rho n \log(\rho r))$ time, 
which subsumes the $O(r\log(\rho r))$ cost to 
replace $W_u$ by $W_{u'}$ and $W_{u''}$ in the parent of $u$. 

\medskip

Note that the insertion of a single point could produce one split per level of
the B-tree. To avoid this, we use a deamortization technique by Fleischer
\cite{Fle96}. This maintains an $(\alpha,2\beta)$-tree (for $\alpha \le 2\beta$)
storing $n$ keys in the leaves, and each leaf is a bucket storing at
most $2\log_\alpha n$ keys. It supports constant-time insertion and deletion of
a key once its location in a leaf is known, guaranteeing at most one split per
insertion or deletion. The premises are consistent with our setting, with
$\alpha=\beta=\rho$, and storing $O(r) = o(\log n)$ keys per leaf.

\subsubsection{Deletions}

Deletion of a point $(x,y)$ starts by searching the B-tree for the 
$x$-coordinate $x$. The point will be found in its leaf, and also possibly
in some cell $W_v[y]$ of some internal node $v$. The search takes
$O(\log n)$ time because, for internal nodes $u$, we binary search the
coordinates $x(v_i)$ stored in $u$ for the correct child $v$, in $O(\log\rho)$ 
time, and then only have to check if $W_v[y].x = x$. In leaves $l$, we binary 
search for $x$ in $\Yx_l'$ in $O(\log r)$ time.

If the point has to be deleted from some $W_v[y]$, we carry out the
{\em uninsertion} process already described, in $O(\log_\rho n \log(\rho r))$
time. We also remove the point $(x,y)$ itself from leaf $l$. 
When a leaf $l$ underflows, we merge it with a neighbor leaf and, if necessary,
split it again. The merging process is analogous to the splitting and can be
easily carried out in $O(r)$ time, plus $O(r\log(\rho r))$ to update the
structures $\Yx$ and $\Yw$ in the parent.

If an internal node underflows, we also merge it with its neighbor and 
re-split it if necessary. The merging of two sibling nodes $v$ and $v'$
is carried out in $O(\rho)$ time, including the construction of the $\Yx$
and $\Yw$ structures for the merged node, $u$. The difficult part is, again, to
get rid of the arrays $W_v$ and $W_{v'}$ at the parent node, replacing them
by a new $W_{v^*}$ table for the merged node $v^*$. For this sake, we choose the
maximum weight between each $W_v[y]$ and $W_{v'}[y]$ and assign it to 
$W_{v^*}[y]$.
The point that was not chosen among $W_v[y]$ and $W_{v'}[y]$ must be 
{\em reinserted}, as before.
Finally, we must rebuild the $\Yx$ and $\Yw$ structures of the parent of $v^*$.
The total cost is $O(r\log_\rho n\log(\rho r))$, just as for insertions.

Note that, upon leaf or internal node merges, a separating label $x(v)$ becomes
unused, and it is removed from $X$. Again, Fleischer's technique \cite{Fle96}
ensures at most one underflow per update.

\subsubsection{Queries}

\paragraph{Identifying the relevant nodes.}
To solve a top-$k$ query with label restriction $[a,b)$ and $y$-coordinate 
restriction $[0,y)$ on the slice, we first identify the $O(\log_\rho n)$ ranges
of siblings of the B-tree tree that exactly cover the interval of labels 
$[a,b)$; plus up to 2 leaf nodes that partially overlap the interval. 
For each node $u$ that is the parent of a range of children $v_s,\ldots,v_e$
included in the cover, we find the maximum weight in $W_{v_s}[0,y-1],\ldots,
W_{v_e}[0,y-1]$ and insert the result in a max-priority 
queue $\cQ$ sorted by the weights of the points. 
Such maximum across $W_{v_i}[0,y-1]$ arrays is obtained in constant time using
universal tables on $\Yw_u$. If $l$ is a leaf partially covering $[a,b)$, 
then the interval is some $W_l'[0,y']$, where $y'=Y_l[y]$.
In addition, we must binary search $\Yx_l'$ for the range $[x_a,x_b]$ 
corresponding to the interval $[a,b)$.
Furthermore, we can only return points whose $Q_l$ bit is not set, to avoid 
repeated answers. Knowing the range in $\Yx_l'[x_a,x_b]$ and the range
$W_l'[0,y']$, the
maximum weight can be obtained from $\Yx_l'$, $\Yw_l'$ and $Q_l$ with a 
universal table, in constant time. 
Identifying the cover nodes and finding their $O(\log_\rho n)$
maxima takes $O(\log n)$ time, and leaves add only $O(\log r)$ time.

Each element inserted in $\cQ$ coming from a range of siblings will be a tuple 
$(u,s,e,i,z,k)$, where $u$ is the parent node of the range of children
$v_s,\ldots,v_e$ in the cover, $(i,z)$ means that the maximum was found at 
$W_{v_i}[z]$ ($s \le i \le e$), 
and $k$ indicates that the point $W_{v_i}[z]$ is the $k$th in the 
range of interest for $u$. All the nodes initially inserted have $k=1$.

The elements inserted in $\cQ$ coming from leaves $l$ are of the form
$[l,j,x_a,x_b,k]$, meaning that the maximum was found in $W_l'[j]$, that
the range of interest is $W_l'[0,Y_l[y]]$ and $\Yx'[x_a,x_b]$, and that the 
point is the $k$th in the range of interest. The first insertions use $k=1$.

We also insert in $\cQ$ a third kind of tuples, namely, the maximum-weight 
point in $W_v[0,y-1]$ with $x$-coordinate in $[a,b)$, for each of the 
$O(\log_\rho n)$ ancestors $v$ of the cover nodes, as they may also hold 
relevant points. To find those maxima we consider the parent $u$ of $v$ and
binary search $\Yx_u$ for $a$ and $b$, to find a mapped interval 
$\Yx_u[x_a,x_b]$, in $O(\log(\rho r))$ time.
Note that this area of $\Yx_u$ corresponds to nodes in $W_v$. Then we use
universal tables on $\Yx_u$ and $\Yw_u$ to find the maximum weight of
$y$-coordinate below $y$ and in the range $\Yx_u[x_a,x_b]$.
For these nodes we insert tuples of the form
$\langle u,v,z,x_a,x_b,k\rangle$ in $\cQ$, meaning that the maximum was obtained
from $W_v[z]$, the range of interest is $\Yx_u[x_a,x_b]$, and the point is the
$k$th in its range of interest. Therefore the initial computation on these
nodes requires $O(\log_\rho n \log(\rho r))$ time. 
Recall that the root node of the B-tree will also
have a $W$ structure computed (this is easily treated as a special case).

We implement $\cQ$ as a Thorup's priority queue \cite{Tho04} on the 
universe of weights $[1,O(n)]$. Note that we do not need to insert the whole 
initial set of $O(\log_\rho n)$ tuples in $\cQ$ if this number exceeds $k$:
if a tuple is not among the first $k$, it cannot contribute to the answer.
Then we use linear-time selection to find the $k$th largest weight in the
tuples and then insert only the first $k$ tuples in $\cQ$. This structure
supports insertions in constant time, thus
the initialization of $\cQ$ takes time $O(\log_\rho n)$.

\paragraph{Extracting the top-$k$ points.}
The first answer to the top-$k$ query is among the $O(k)$ tuples
we have inserted in $\cQ$. Therefore, to obtain the first result, we extract 
the tuple with maximum weight from $\cQ$. If it is of the form
$[l,j,x_a,x_b,k]$, that is, it comes from a leaf $l$, we report the
point $W_l'[j]$, compute the $(k+1)$th highest-weight point $W_l'[j']$
within $W_l'[0,Y_l[y]]$ and $\Yx_l'[x_a,x_b]$ using universal tables, and 
reinsert tuple $[l,j',x_a,x_b,k+1]$ in $\cQ$. If, instead, the maximum tuple
extracted from $\cQ$ is of the form
$\langle u,v,z,x_a,x_b,k\rangle$, that is, it becomes from an ancestor of a
cover node, we report the point $W_v[z]$, compute the $(k+1)$th highest-weight 
point $W_v[z']$ with $y$-coordinate below $y$ and within $\Yx_u[x_a,x_b]$
using universal tables, and reinsert tuple
$\langle u,v,z',x_a,x_b,k+1\rangle$. Finally, if the maximum tuple extracted
from $\cQ$ is of the form $(u,s,e,i,z,k)$, we report the point
$W_{v_i}[z]$, where $v_i$ is the $i$th child of $u$, compute the $(k+1)$th 
highest-weight point $W_{v_{i'}}[z']$ in 
$W_{v_s}[0,y-1],\ldots,W_{v_e}[0,y-1]$ using universal tables, and reinsert 
tuple $(u,s,e,i',z',k+1)$. If the extracted point had $k=1$, however, it is
possible that the next highest-weight element comes from the child $v_i$.
Therefore, if $v_i$ is an internal node, we compute the highest-weight point 
in $\Yw_{v_i}$ with $y$-coordinate below $y$. 
Let it be the pair $(i'',z'')$, then we insert a new tuple 
$(v_i,1,a(v_i),i'',z'',1)$ in $\cQ$. If, instead, $v_i$ is a leaf $l=v_i$ with
$r(l)$ elements, then we find the maximum-weight point $W_l'[j']$ in 
$W_{l}'[0,Y_l[y]]$ using $\Yw_{l}'$, and insert the tuple
$[l,j',1,r(l),1]$ in $\cQ$. In all cases the cost to compute and insert the
new tuples is constant.

\medskip

If we carry out $k$ extractions from $\cQ$, we will also carry out up to $2k$ 
insertions, thus the size of $\cQ$ will be $O(k)$ and minima extractions will
cost $O(\log\log k)$ \cite{Tho04}. The cost of this part is then 
$O(k\log\log k)$, and the total query time is
$O(\log n + \log r + \log_\rho n \log(\rho r) + \log_\rho n + k\log\log k)$. 
Given a constant $0 < \eps < 1/2$, we will choose $r = \rho = \lg^\eps n$,
fulfilling the promise that $\rho r = o(\log n)$.
Therefore, the update cost becomes
$O(r\log_\rho n\log(\rho r)) = O(\log^{1+\eps} n)$, and the query time
becomes $O(\log n + k\log\log k)$. 
Note that the process is not online: We must know $k$ in advance so as to 
initially limit the size of 
$\cQ$ to $k$. We use the technique of Section~\ref{sec:online} to make the
process online in $k$. That is, $k$ is not specified in advance and the process
can be interrupted after having produced any number $k$ of results, and the
total cost paid will be $O(\log n + \log\log k)$.

\subsection{Multiresolution Grids} \label{sec:multires}

We extend the result of Section~\ref{sec:slimgrid} to grids of polylogarithmic
height $r^c$, for some constant $c$. We will represent the grid at various 
resolutions and split it into slim grids for each resolution. 
Consider a virtual perfect tree of arity $r$ and $n$ leaves, so that the
$i$th left-ro-right node of height $j$ covers the rows $(i-1)\cdot r^j+1$
to $i\cdot r^j$. The tree is of height $c$.

For each node $v$ of this tree we store a slim grid of $r$ rows, one per child.
All the points whose row belongs to the area covered by the $i$th child of $v$ 
will be represented as having $y$-coordinate $i$ in the slim grid of $v$. 

When a new point $(x,y,w,d)$ is inserted in the grid, we insert it into the
$c$ slim grids that cover it, giving it the appropriate row value in each
slim grid, and similarly when a point is deleted. The $x$-coordinate labels
are shared among all the grids. This arrangement multiplies space and insertion
and deletion times by the constant $c$.

Now consider a 3-sided top-$k$ query with the restriction $[a,b)$ on the
$x$-coordinates and $[0,y)$ on the $y$-coordinates. The range $[0,y)$ is
covered with the union of one range in one slim grid per level of the tree.
Let $y_c, \ldots, y_1$ be the child numbers of the path from the root to
the $y$-th row of the grid. Then we take the 3-sided query $[a,b) \times
[0,y_j)$ at the node of height $j$ in the path.

We start the searches in the $c$ slim grids, and extract the first result
from each grid. We insert those local maxima into a new global queue $\cQ$. 
Now we repeat $k$ times the process of extracting the next result from $\cQ$, 
reporting it, requesting the next result from the grid where the result came
from, and inserting it in $\cQ$. Note that $\cQ$ can be implemented naively
because it contains at most $c$ elements and $c$ is a constant.

Initializing the searches will then require $O(c \log n)$ 
time, and extracting $k$ results from the slim grids will require 
$O(k\log\log k)$ time. Managing $\cQ$ will require $O(ck)$ time
even if done naively. Therefore the total time is still $O(\log n
+ k\log\log k)$. The update time per element stays 
$O(c \log^{1+\eps} n)$. The process is also online. 
Then we obtain the following lemma.

\begin{lemma}\label{lem:dynweight}
  A set of $n$ points, one per column on an $n\times r^c$ grid, for 
  $r=\lg^\eps n$ and any constant $0<\eps<1$ and $c\ge 1$,
  with weights in $[1,O(n)]$, can be stored in
  $O(n)$ words of space, so that for any $1\leq k \le n$, $1 \le h\leq r^c$ and
  $1\leq a\leq b\leq n$, $k$ most highly weighted points in the range
  $[a,b]\times [0,h]$ can be reported in 
  decreasing order of their weights in $O(\log n + k\log\log k)$
  time, online in $k$. Points (and their columns) can be inserted and deleted 
  in $O(\log^{1+\eps} n)$ time.
\end{lemma}

\subsection{The Final Result} \label{sec:dynfinal}

We find the locus $v$ of $P$ in the suffix tree in time 
$O(p(\log\log n)^2/\log_\sigma n+\log n)$. 
Then the $x$-coordinate range of labels to search for in the grid is $[a,b)$,
where $a$ is the first label in $\mathit{list}(v)$ and $b =\mathit{last}(v)$.
Since we store string depths in the grid, the 
$y$-coordinate range of the query is $[0,p)$.

Our dynamic grid is horizontally split into bands of $r^c$ rows, for a constant
$c$, which are handled as explained in Section~\ref{sec:multires}. 
Therefore, our 3-sided query is translated into 3-sided queries on the first 
$\lceil p/r^c\rceil$ bands. All but the last will query for the whole row
interval $[0,r^c)$, whereas the latter will query for the row interval
$[0,(p-1)~\mathrm{mod}~r^c]$.

We start the searches in all the bands, and extract the first result
from each. If there are more than $k$ bands, we use linear-time selection to
keep only the $k$ highest weights. Then we insert the local maxima into a new 
global queue $\cQ$. 
Now we repeat $k$ times the process of extracting the first result from $\cQ$, 
and if it came from the $i$th band, then we request the next result from that
band and insert it in $\cQ$ (unless it has no more results, in which case we
continue with the remaining bands). Again, $\cQ$ will be implemented with
Thorup's priority queue \cite{Tho04}.

Initializing the searches will then require $O(\lceil p/r^c\rceil \log n)$ time,
and extracting $k$ (and inserting other $k$) results in $\cQ$ will take time 
$O(k\log\log k)$. We choose $r=\lg^\eps n$ for some $0<\eps<1/2$, as explained,
and rename $c$ as $(c+1)/\eps$. Therefore, we
obtain a query time of $O(p/\log^c n + \log n + k\log\log k)$ for the grid.
Once again, the scheme can be made online with the technique of
Section~\ref{sec:online}.
Updating grid points, including extending the grid downwards, requires 
$O(\log^{1+\eps} n)$ time. This yields Theorem~\ref{theor:dynweight}. 

We developed our result for $\tf$ as the relevance measure.
It is very easy to support others like $\docrank$, but if the weights are 
not integer numbers, then Thorup's priority queues \cite{Tho04} cannot be
used. In this case we insert all the new weights that appear in a data structure
for monotonic list labeling, which assigns them integers in a polynomial
universe $[1,n^{O(1)}]$. 
This adds at most $O(\log n)$ time per symbol inserted \cite{DR93} 
(deletions can be handled by deamortized periodic rebuildings).
In general we can support any
measure that can be computed in time $O(|d|C_w)$ over the suffix tree of the
document to insert: We explicitly build such suffix tree, compute the relevance
measure for all the nodes, and then use them to assign the weights as we insert
the nodes in our suffix tree. At the end we delete the suffix tree of $d$. This
suffix tree can be built (and deleted) in $O(|d|)$ time on integer alphabets
\cite{Far97}. Therefore we simply charge $O(C_w)$ time per character inserted
in our text collection. Note that $C_w$ is $O(1)$ for measures $\tf$ and
$\docrank$. Hon et al.~\cite{HSV09} show how to compute $\mindist$ from
the suffix tree of the document in $O(|d|\log |d|)$ time, so $C_w=O(\log n)$
in this case.

Our scheme works as long as $\lg n$ has a fixed value (plus $O(1)$). We use
standard techniques to incrementally rebuild the structure for larger or smaller
$\lg n$ values as more insertions or deletions are processed.

\section{A Space-Efficient Data Structure}
\label{sec:freqspace}

We show now how the space of our static structure can be reduced to
$O(n(\log\sigma + \log D + \log\log n))$ bits, where $\sigma$ is the 
alphabet size and $D$ is the number of documents, and retain almost the same 
query time. Our approach is to partition the tree into minitrees, which are
represented using narrower grids.

\paragraph{Partitioning the tree.} We define $z = \Theta(\sigma D \log n)$.
We say that a node $v\in T$ is {\em heavy} if the subtree rooted at $v$ 
has at least $z$ leaves, otherwise it is {\em light}.
A heavy node is \emph{fat} if it has at least two heavy children, otherwise
it is \emph{thin}. 

All the non-fat nodes of $T$ are grouped into minitrees as follows. 
We traverse $T$ in depth-first order. If a visited node $v$ has two heavy
children, we mark $v$ as fat and proceed.
If $v$ has no heavy children, we mark $v$ as thin or light, and
make $v$ the root of a minitree $T_v$ that contains all the
descendants of $v$ (which need not be traversed). 
Finally, if $v$ has one heavy child $v_1$, we mark $v$
as thin and make it the root of a minitree $T_v$. The extent of this minitree
is computed as follows. If $v_i$, $i\geq 1$, is a thin node with one heavy
child $v_{i+1}$, we visit nodes $v_1,v_2,\ldots v_{j-1}$ and include $v_i$ 
and all the descendants of its other children, until either $v_{j-1}$ has no
heavy children or is fat, or $T_v$ contains more than $\sigma z$ nodes after 
considering $v_j$. Then we continue our tree traversal from $v_j$. Note that
$T_v$ contains at the very least the descendants of $v$ by children other than 
$v_1$.

With the procedure for grouping nodes described above, the leaves of 
minitrees can be parents 
of nodes not in the minitree. Those child nodes can be either fat nodes or
roots of other minitrees. However, at most one leaf of a minitree 
can have children in $T$.

Note that  the size of a minitree is at most $O(\sigma z)$. On the other
 hand, as 
two heavy children have disjoint leaves, there are $O(n/z)$ fat nodes in $T$. 
Finally, minitrees can contain as little as one node (e.g., for leaves that
are children of fat nodes). However, note that a minitree root is either a 
child of a fat node (and thus there are $O(\sigma n/z)$ minitrees of this
kind), or a child of a leaf of another minitree such that 
the sum of both minitree sizes exceeds $\sigma z$ (otherwise we would have
included the root $v_j$ of the child minitree as part of the parent minitree).
Moreover, as said, at most one of the leaves of a minitree can be the parent 
of another minitree, so these minitrees that are ``children'' of others form
chains where two consecutive minitrees  cover at least $\sigma z$ nodes of $T$. 
Thus there are $O(n/(\sigma z))$ minitrees of this second kind. Adding up both
cases, there are $O(\sigma n/z) = O(n/(D\log n))$ minitrees in $T$.

\paragraph{Contracted tree and minitrees.}
The pointers in a tree $T$ are defined in the same way as in 
Section~\ref{sec:topkfram}.
Since we cannot store $T$ without violating the desired space bound, 
we store a {\em contracted tree $T^c$} and the minitrees $T_v$.

The contracted tree $T^c$ contains all fat nodes of $T$, plus one node
$v^c$ for each minitree $T_v$. Each pointer $\ptr(u,d)=u'$ of $T$ is
mapped to a pointer $\ptr^c(u^c,d)=(u')^c$ of $T^c$ as follows. If $u$ 
is a fat node, then $u^c=u$. Else, if $u$ belongs to minitree $T_v$, then
$u^c = v^c$. Similarly, if $u'$ is a fat node then $(u')^c = u'$; else
if $u'$ belongs to minitree $T_{v'}$ then $(u')^c = (v')^c$. In other words,
nodes of a minitree are mapped to  the single node
that represents that minitree in $T^c$ and pointers are changed 
accordingly.

For each minitree $T_v$, we store one additional dummy node $\nu$ that 
is the parent of $v$. If a leaf $u_h$ of $T_v$ has a heavy child 
$u'\not\in T_v$, we store an additional dummy node $\nu' \in T_v$ that is the only
child of $u_h$. Pointers of $T_v$ are modified as follows. 
Each pointer $\ptr(u,d)$, $u\in T_v$, that points to an ancestor 
of $v$ is transformed into a pointer $\ptr(u,d)$ that points to $\nu$. 
Every pointer $\ptr(u'',d)$ that starts in a descendant $u''$ of $u_h$ 
and points to a node $u\in T_v$, $u\not=u_h$,
 (respectively to an ancestor of $v$) 
is transformed into $\ptr(\nu',d)$ that starts in $\nu'$ and points to $u$ 
(respectively to $\nu$).   
By Lemma~\ref{lemma:uniqptr}, there are at most $D$  such pointers $\ptr(u'',d)$.
We observe that there is no need to store pointers to the node $u_h$ 
in the minitree $T_v$ because such 
pointers are only relevant for the descendants of $u_h$ that do not 
belong to $T_v$.
\paragraph{Suffix trees.}
The contracted tree $T^c$ consists of $O(n/(D\log n))$ nodes, and thus it would 
require just $O(n/D)$ bits. The minitrees contain $O(\sigma z)$ nodes,
but still an edge of a minitree can be labeled with a string of length 
 $\Theta(n)$. Instead of representing the
contracted tree and the minitrees separately, we use Sadakane's compressed 
suffix tree (CST) \cite{Sad07} 
to represent the topology of the whole $T$ in $O(n)$ bits, and a 
compressed representation \cite[Thm.~5.3]{GGV03} of the global suffix array 
(SA) of the string collection, which takes $O(n\log\sigma)$ bits. This SA 
representation finds the suffix array interval $[l,r]$ of $P$ in time 
$O(p/\log_\sigma n + \log^\eps n\log\sigma)$ for any constant $\eps>0$, and a 
lowest-common-ancestor query for the $l$-th and $r$-th leaves of $T$ finds the
locus $u$ of $P$ in $O(1)$ additional time. A bitmap $M[1,n]$ marks which
nodes are minitree roots, and another bitmap $C[1,n]$ marks which nodes are
fat or minitree roots. Both are indexed with preorder numbers of $T$, which
are computed in constant time on the CST. With a simple $O(n)$-bit structure
for constant-time marked ancestor queries that is compatible with our CST 
representation \cite[Sec.~4.1]{RNO11}, we can find the lowest ancestor $v$ of 
$u$ marked in $M$ or in $C$.
With bitmap $M$ we can identify whether $u$ belongs to a minitree rooted at 
$v$ (with local preorder $preorder_{T_v}(u) = preorder_T(u)-preorder_T(v)$ and 
depth $depth_{T_v}(u) = depth_T(u)-depth_T(v)$; depths are also computed in 
constant time). Similarly, with $C$ and $M$ we can 
identify whether $u$ is a fat node, and find out its preorder in $T^c$ as 
$preorder_{T^c}(u) = rank_1(C,preorder_T(u))$, in constant time. Its depth in 
$T^c$ can be stored in an array indexed by $preorder_{T^c}$ in $O(n/D)$ bits.

\paragraph{Contracted grid.}

We define the grid of the contracted tree $T^c$ as in 
Section~\ref{sec:topkfram}, considering all pointers $\ptr^c$.
Those are either $\ptr$ pointers leaving from fat nodes, or
leaving from inside some minitree $T_v$ and pointing above $v$.
For every fat node and for every minitree $T_v$, and for each document $d$,
there is at most one such pointer by Lemma~\ref{lemma:uniqptr}. Thus each 
node of $T^c$ contributes at most $D$ pointers $\ptr^c$. As there are 
$O(n/(D\log n))$ nodes, there are $O(n/\log n)$ pointers $\ptr^c$ in 
$T^c$.

Therefore, the grid associated to $T^c$ is of width $O(n/\log n)$ and height
$O(n/(D\log n))$. As there are $O(n/\log n)$ distinct weights among the $\ptr^c$
pointers, we only store their ranks. This change does not alter the result of 
any top-$k$ query. Therefore the data structure of Theorem~\ref{theor:weight} 
on $T^c$ occupies $O(n/\log n)$ words, or $O(n)$ bits.

\paragraph{Local grids.}

The local grid for a minitree $T_v$ collects the pointers $\ptr$ local to
$T_v$. It also includes at most $D$ pointers towards its dummy root $\nu$,
and at most $D$ pointers coming from its node $\nu'$, if it has one. Overall
$T_v$ contains $O(\sigma z)$ pointers and $O(\sigma z)$ nodes, so its grid
is of size $O(\sigma z) \times O(\sigma z)$. The weights are also replaced by 
their ranks, so they are also in the range $[1,O(\sigma z)]$. Using
Theorem~\ref{theor:weight} the minitree requires $O(\log(\sigma z))$ bits
per node. Added over all the nodes of $T$ that can be inside minitrees, the
total space is $O(n\log(\sigma z)) = O(n(\log \sigma + \log D + \log\log n))$.
Note that the tree topology is already stored in the CST, so information 
associated to nodes $u\in T_v$ such as the  intervals $[l_u,r_u]$ can be
stored in arrays indexed by preorder numbers.

\paragraph{Queries.}
Given a query pattern $P$, we find the locus $u$ of $P$ and determine
whether $u$ is a fat node or it belongs to a  minitree in $O(p/\log_\sigma n
+ \log^\eps n \log \sigma)$ time,
as explained. If $u$ is fat, we solve the query on the contracted grid of
$T^c$. Note that this grid does not distinguish among different 
nodes in the same minitree.
But since $u$ is an ancestor either of all 
nodes in a minitree or of none of them, such distinction is not necessary.

If $u$ belongs to a minitree $T_v$, we answer the query using the corresponding
local grid. This grid does not distinguish where exactly  the 
pointers pointing to $\nu$ lead, nor where exactly the pointers that
originate in $\nu'$ come from. 
Once again, however, this information is not important 
in the case where the locus $u$ of $P$ belongs to $T_v$.

Note that we still need to maintain the global array mapping $x$-coordinates
to document identifiers. This requires $O(n\log D)$ bits.

\begin{theorem}\label{theor:topkspace}
  Let $\mathcal{D}$ be a collection of  $D$ documents over an integer 
  alphabet $[1,\sigma]$ with  
  total length $n$, and  let $w(S,d)$ be a function  that assigns a
  numeric  weight
  to string $S$ in document $d$, that depends only on the set of 
  starting positions of occurrences of $S$ in $d$. 
  Then there exists an $O(n(\log D+\log \sigma +\log \log n))$-bit 
  data structure that, given a string
  $P$ and an integer $k$, reports $k$ documents $d$ containing $P$ with highest
  $w(P,d)$ values, in decreasing order of $w(P,d)$, in $O(p/\log_\sigma n +
\log^\eps n \log\sigma +k)$ time, for any constant $\eps>0$.
\end{theorem}

In case $p < \lg^{1+\eps} n$, we can use a different compressed suffix array
\cite{BN11}, which gives $O(p)$ search time, and the overall time becomes
$O(p+k)$.

\subsection{A Smaller Structure when using Term Frequencies}

In this section we show that  
the space usage can be further improved if $w(P,d)=\tf$, 
i.e., when the data structure must report $k$ documents in which $P$ 
occurs most frequently. 

Our improvement is based on applying the approach of 
Theorem~\ref{theor:topkspace} to each minitree. 
The nodes of a minitree are grouped into microtrees; 
if the structure for a microtree still needs too much space, 
we store them in a compact form that will be described below. 

Let $z'= \sigma D\log m$, where $m$ is the number of nodes 
in a  minitree $\cT$. 
Using the same method as in Theorem~\ref{theor:topkspace}, we divide 
the nodes of $\cT$ into $O(m/z')$ minifat nodes and 
$O(m/(D\log m))$ microtrees, so that each microtree contains 
$O(\sigma z')$  nodes. 
We construct the contracted minitree and the contracted grid for 
$\cT$ as in Theorem~\ref{theor:topkspace}. 
Both the contracted minitree and the structure for the 
contracted grid use $O(m)$ bits.
We can traverse a path in the microtree using the implementation of 
the global suffix tree described in the previous section, as well as
compute local preorders and depths, and attach satellite information to
microtree nodes.

For every microtree $\cT_v$, we define the dummy nodes 
$\nu$ and $\nu'$. Pointers in $\cT_v$ are transformed 
as in the proof of Theorem~\ref{theor:topkspace} with regard to $\nu$ and
$\nu'$.

Let $m'$ denote the number of nodes in a microtree.
If $\log m' = O(\log\sigma + \log D)$, we implement 
the local grid data structure described in Theorem~\ref{theor:topkspace} 
for a microtree. 
In this case we can store a data structure for a microgrid in 
$O(\log(m'+D))=O(\log \sigma + \log D)$ bits per node.

If, instead, $\log m' = \omega(\log\sigma + \log D)$, since 
$\log m' = O(\log(\sigma z')) = O(\log\sigma + \log D + \log\log m)$,
it follows that $\log m' = O(\log\log m)$. Hence, the size of 
the microtree is $m' = \log^{O(1)} m = (\log\sigma+\log D+\log\log n)^{O(1)}
= (\log\log n)^{O(1)}$. The total number of pointers in the microtree is
also $m'' = m'+O(D) = (\log\log n)^{O(1)}$ (since $\log D = o(\log m')$).
Since all the grids in $m'' \times m'$, with one point per $x$-coordinate, and 
weights in $[1,m'']$, can be expressed in $m''(\log m' + \log m'') =
o(\log n)$ bits, we can store pre-computed 
answers for all possible queries on all possible small microtrees.  
The only technical difficulty is that weights of some pointers in a microtree
can be arbitrarily large. However, as explained below, 
it is not necessary to know the exact weights of pointers to answer 
a query on a small microtree. 

All pointers $\ptr(u_l,d)$ where $u_l$ is a leaf node and 
$u_l\not=\nu'$ have weight $1$. The weights of $\ptr(\nu',d)$ can  be 
arbitrarily large. The weight of a pointer $\ptr(u,d)$ for an internal node $u$ 
equals to the sum of weights of all pointers $\ptr(u',d)$ for the same 
document $d$ that lead to $u$. Thus the weight of $\ptr(u,d)$
can also be large. 
We note that there is at most one pointer $\ptr(\nu',d)$ for each $d$.
Therefore the weight of each pointer $\ptr(u,d)$ can 
be expressed as the sum $w_1(u)+w_2(u)$, where $w_1(u)$ is the weight of 
$\ptr(\nu',d)$ or $0$ and $w_2(u)\leq m'$. 
In other words, the weight of $\ptr(u,d)$ differs from the weight of 
$\ptr(\nu',d)$ by at most $m'$. 

\tolerance=1000
Let the set $\cN$ contain the weights of all pointers $\ptr(u_l,d)$ and 
$\ptr(\nu',d)$. Let $\cN'=\{\, \floor{w/m'}, \, \floor{w/m'}+1 \,|\, w\in \cN \}$.
To compare the weights of any two pointers it is sufficient to know 
(i) the tree topology 
(ii) for every leaf $u_l$, the document $d$ whose suffix is stored in $u_l$ 
(iii) for every $\ptr(\nu',d)$, the pair $(\rank(\floor{w/m'},\cN'), w \mod m')$
where $w$ is the weight of $\ptr(\nu',d)$. 
There are $o(n/\log n)$ possible combinations of tree topologies 
and possible pairs $(\rank(\floor{w/m'},\cN'), w \mod m')$. 
Hence, we can store answers to all possible queries for all microtrees
in a global look-up table of size $o(n)$ bits.

The topology of a microtree can be stored in $O(m')$ bits.
We can specify the index of the document $d$ stored in a leaf $u_l$ 
with $\log D$ bits. 
We can specify each pair $(\rank(\floor{w/m'},\cN'), w \mod m')$ 
with $O(\log m')$ bits. Since $D=O(m'/\log m')$, information from 
item (iii) can be stored in $O(m')$ bits. 
Thus each microtree can be stored in $O(m'\log D)$ bits if 
$\log m' = \omega(\log\sigma+\log D)$.
Summing up, our data for a minitree uses $O(m(\log\sigma+\log D))$ bits. 
Therefore the total space usage is $O(n(\log\sigma + \log D))$ bits.  

A query for a pattern $P$ is answered by locating the locus $u$ 
of $P$. If $u$ is a fat node in $T$, the query is answered 
by a data structure for the contracted grid. 
If $u$ belongs to a minitree $\cT$ and $u$ is a minifat node, 
we answer the query by employing the data structure for 
the contracted grid of $\cT$. If $u$ belongs to a microtree 
$\cT_v$, the query is answered either by a microgrid data structure 
or by a table look-up.   
\begin{theorem}\label{theor:freqspace}
  Let $\mathcal{D}$ be a collection of  strings over an integer alphabet 
  $[1,\sigma]$ with  
  total length $n$, and  let $\tf(P,d)$ denote the number of occurrences of $P$ 
  in $d$.
  Then there exists an $O(n(\log D+\log \sigma))$ bit 
  data structure that, given a string
  $P$ and an integer $k$, reports $k$ documents $d$ containing $P$ with highest
  $\tf(P,d)$ values, in decreasing order of $\tf(P,d)$, in $O(p/\log_\sigma n +
\log^\eps n \log\sigma +k)$ time, for any constant $\eps>0$.
\end{theorem}

\section{Parameterized Top-$k$ Queries}
\label{sec:topk2dim}

In this section we improve a recent data structure that supports two-dimensional
top-$k$ queries \cite[Sec. 5]{NNR13}. The structure is similar to our wavelet 
tree $W$ described in the proof of Lemma~\ref{lemma:count1}. In addition, for 
the points stored at any node of $W$, it stores an RMQ data structure that gives
in constant time the position of the point with maximum weight within any 
interval. As explained, this structure \cite{FH11} uses $O(t)$ bits if the node 
of $W$ handles $t$ points, and thus the total space of this extended wavelet
tree $W$ is $O(n)$ words for an $O(n) \times O(n)$ grid.

They \cite{NNR13}
show how to support top-$k$ queries in a general interval $[a,b] \times
[c,d]$ by first identifying the $O(\log n)$ nodes $v \in W$ that cover $[c,d]$,
mapping the interval $[a,b]$ to $[a_v,b_v]$ in all those nodes $v$, and setting
up a priority queue with the maximum-weight point of each such interval. Now, 
they repeat $k$ times the following steps: $(i)$ extract the maximum weight 
from the queue and report it; $(ii)$ replace the extracted point, say 
$x \in [a_v,b_v]$, by two points corresponding to $[a_v,x-1]$ and $[x+1,b_v]$,
prioritized by the maximum weight in those ranges.

Their total time is $O((k+\log n)\log n)$ if using linear space. The $O(\log n)$
extra factor is due to the need to traverse $W$ in order to find out the real
weights, so as to compare weights from different nodes. However, those weights
can be computed in time $O(\log^{\eps} n)$ and using $O(n\log n)$ extra bits
\cite{Cha88,Nek09,CLP11}. The operations on the priority queue can be carried out
in $O(\log \log n)$ time~\cite{Tho04}. Thus we have the following result.

\begin{lemma}\label{lemma:topk2d}
Given a grid of $n \times n$ points,
there exists a data structure that uses $O(n)$ words of space and 
reports $k$ most highly weighted points in a range $Q=[a,b]\times [c,d]$ 
in $O((k+\log n)\log^\eps n)$ time, for any constant $\eps>0$. 
The structure is built in $O(n\log n)$ time.
\end{lemma}

Note this technique automatically admits being used in online mode (i.e.,
without knowing $k$ in advance), since we have not made use of $k$ to speed
up the priority queue as in previous sections. 
We can easily stop the computation at some $k$ and resume it later.

\subsection{Limited Three-Dimensional Queries}
\label{sec:exttopk2dim}

In this section we slightly extend the scenario considered above.
We assume that each point has an additional coordinate, denoted $z$, 
and that $p.z\leq \log^{\alpha} n$ for a constant $\alpha>0$.
Top-$k$ points in a three-dimensional range 
$[a,b]\times [c,d]\times [\beta,\gamma]$ must be reported 
sorted by their weights. Such queries will be further called 
{\em limited three-dimensional top-$k$ queries}. 
We can obtain the same result as in Lemma~\ref{lemma:topk2d} for these
queries.

Instead of a binary wavelet tree, we use a multiary one \cite{FMMN07},
with node degree $\log^{\eps} n$ and height $O(\log n/\log \log n)$. Now
each node $v \in W$ has associated a vector $B_v$ so that $B_v[i]$
contains the index of the child in which the $i$-th point of $v$ is stored. 
Using $B_v$ and some auxiliary data structures, we can obtain 
the weight of any point at any node in $O(\log^\eps n)$ time \cite{Nek09}. 
All vectors $B_v$ and the extra data structures use $O(n)$ words. 

We regard the $t$ points of each node $v$ as lying in a two-dimensional
grid of $x$- and $z$-coordinates. Instead of one-dimensional RMQs
on the $x$-coordinates $[a_v,b_v]$, we issue two-dimensional RMQs
on $[a_v,b_v] \times [\beta,\gamma]$. The wavelet tree of the basic 
two-dimensional RMQ data structure \cite{NNR13} handles $n \times m$ grids 
in $O(n\log m)$ bits of space and answers RMQs in time $O(\log^2 m)$.
In our case $m < \log^\alpha n$ and thus the space is $O(n\log\log n)$ bits
and the query time is $O((\log\log n)^2)$. Thus the space of the two-dimensional
data structures is of the same order of that used for vectors $B_v$, adding
up to $O(n)$ words. As RMQs are built in linear time, the construction time
is $O(n\log n)$.

Now we carry out a procedure similar to that of the two-dimensional version. 
The range $[a,b]$ is covered by $O(\log^{1+\eps} n/\log\log n)$ nodes. We obtain
all their (two-dimensional) range maxima, insert them in a priority queue,
and repeat $k$ times the process of extracting the highest weight and replacing
the extracted point $x \in [a_v,b_v]$ by 
the next highest weighted point in $[a_v,b_v]$ (thus we are running these 
range maxima queries in online mode).

The two-dimensional RMQ structures at nodes $v$ cannot store the absolute 
weights within overall linear space. Instead, when they obtain the 
$x$-coordinate of their local grid, this coordinate $x_v$ is mapped to the 
global $x$-coordinate in $O(\log^\eps n)$ time, using the same technique as 
above. Then the global array of weights is used. 
Hence these structures find a two-dimensional maximum weight in time
$O(\log^\eps n (\log\log n)^2)$. This is repeated over
$O(\log^{1+\eps} n / \log\log n)$ nodes, and then iterated $k$ times. The
overall time is $O((k+\log^{1+\eps} n/ \log\log n)\log^\eps n (\log\log n)^2)$,
which is of the form $O((k+\log n)\log^\eps n)$ by adjusting $\eps$.
The times to handle the priority queue are negligible \cite{Tho04}.

\begin{lemma}\label{lemma:exttopk2d}
Given a grid of $n \times n \times \log^\alpha n$ points, for a constant
$\alpha>0$, there exists a data structure that uses $O(n)$ words of space and 
reports $k$ most highly weighted points in a range $Q=[a,b]\times [c,d] \times
[\beta,\gamma]$ in $O((k+\log n)\log^\eps n)$ time, for any constant $\eps>0$. 
It is built in $O(n\log n)$ time.
\end{lemma}

Again, this result holds verbatim in online mode.

\subsection{The Final Result}
\label{sec:parweight}

We divide the grid into horizontal stripes of height 
$r=\ceil{\log^{c+1+\eps} n}$ for any constant $c$, 
much as in Section~\ref{sec:dynfinal}. 
We store a data structure for 
limited three-dimensional top-$k$ queries for each slim grid, taking $y$ as the
limited coordinate. A query $[a,b]\times [0,h]\times [\tau_1,\tau_2]$ is 
processed just as in Section~\ref{sec:dynfinal}, with the only difference that
the queries $[a,b]\times [\tau_1,\tau_2] \times [0,y'] $ to the local grids
now require $O(\log^{1+\eps}n)$ initialization time and then 
$O(\log^\eps n)$ time per element retrieved, according to 
Lemma~\ref{lemma:exttopk2d}. 
Then, we initialize our global query $\cQ$ in time
$O(\lceil h/r \rceil \log^{1+\eps} n) = O(h/\log^c n+\log^{1+\eps}n)$, and
then extract each new result in time $O(\log^{\eps}n)$. The time of the priority
queue is blurred by adjusting $\eps$. Hence, the total query time
is $O(h/\log^c n + (k+\log n)\log^\eps n)$, and Theorem~\ref{theor:parweight}
is proved. 

%
%
%

\section{Conclusions}
\label{sec:concl}

We have presented an optimal-time and linear-space solution to top-$k$
document retrieval, which can be used on a wide class of relevance measures
and subsumes in an elegant and uniform way various previous solutions to
other ranked retrieval problems. We have also presented dynamic variants,
space-reduced indexes, and structures that solve extensions of the basic
problem. The solutions reduce the problem to ranked
retrieval on multidimensional grids, where we also present improved results,
some tailored to this particular application, some of more general interest.

After the publication of the conference version of this article \cite{NN12},
Shah et al.~\cite{SSTV13} showed how to achieve the optimal $O(k)$ time 
once the locus of $P$ is known. This is in contrast to our original
result, where we used time $O(p+k)$ after having spent time $O(p)$ to find the
locus. Their improvement allows one to use these techniques in other scenarios
where the loucs is obtained in some other way, without the need to search for
it directly using $P$. They also extend the results to the important case of
the external memory scenario. The new results we obtain in this article about
how to search the suffix tree in RAM-optimal time $O(p/\log_\sigma n)$, and
how to handle the dynamic scenario, nicely complement those results and add up
to a rather complete solution to the problem.

\medskip

There are several relevant research directions, on which we comment next.

\paragraph{RAM optimality.}
In our previous conference version we had achieved time $O(p+k)$, which was
optimal only in the comparison model (although we used RAM-based techniques). 
Now we have improved this result to $O(p/\log_\sigma n + k)$, which is optimal 
in general in the RAM model (considering the case $\lg D = \Theta(\log n)$), 
because it is the size
in words of the input plus the output. Achieving $O(p/\log_\sigma n)$ time on
the suffix tree, without any polylogarithmic additive penalty, is an 
interesting result by itself, and we have obtained it without altering the
topology of the suffix tree (which is crucial for the invariants of Hon et 
al.~\cite{HSV09} to work). However, we do not know if our solution is optimal 
when there are very few distinct documents, $\lg D = o(\log n)$. The question
of whether $O(p/\log_\sigma n + k/\log_D n)$ time can be achieved is still
open.

\paragraph{Construction time.}
Without considering the cost to compute weights $w(path(u),d)$ for all 
pointers $\ptr$ in the suffix tree, the construction time of Hon et 
al.~\cite{HSV09} (which achieves suboptimal query time) is $O(n)$. The time
to build our grid structure is $O(n\log n)$, to which we must add
$O(n\log^\eps n)$ randomized time to achieve RAM-optimal search time in the
suffix tree traversal (or $O(n\,\mathrm{polylog}~n)$ deterministic time for a
weaker version of it). Is it possible to achieve linear, or at least 
$O(n\log n)$, deterministic construction time for our data 
structures?

\paragraph{Dynamic optimality.}
In our dynamic variant, the static RAM-optimal search time in the suffix tree 
becomes $O(p(\log\log n)^2/\log_\sigma n+\log n)$. There are schemes that do
better for large $\sigma$, for example $O(p+(\log\log\sigma)^2)$ 
time \cite{FG13}. Although they do not support deletions yet, this seems
to be possible.
On the other hand, we obtained $O(\log^{1+\eps} n)$ update time per symbol.
A general question is, which is the best search time we can obtain in the
dynamic scenario?

\paragraph{Practical results.}
Our solutions are not complex to implement and do not make use of impractical
data structures. A common pitfall to practicality, however, is space usage. 
Even achieving linear space (i.e., $O(n\log n)$ bits) can be insufficient. We 
have shown that our structure can use, instead, 
$O(n(\log\sigma + \log D))$ bits for the $\tf$ measure (and slightly more for
others), but the constants are still large.
There is a whole trend of reduced-space representations for
general document retrieval problems with the $\tf$ measure
\cite{Sad07a,VM07,HSV09,CNPT10,GNP11,HST12,BNV13,GKNP13,Tsu13,HSTV13,NT13}. 
The current situation is as follows \cite{Nav13}: One trend aims at the least
space usage. It has managed to use just $D\log(n/D)+O(D)+o(n)$ bits on top of
a compressed suffix array of the collection, and the best time complexity it
has achieved is $O(p+k\log^2 k \log^{1+\eps} n)$ for any constant $\eps>0$
\cite{NT13}. Another trend adds to the space the so-called {\em document array} 
\cite{Mut02}, which uses $n\lg D + o(n\lg D)$ bits and enables faster
solutions. Currently the fastest one achieves 
$O(p+k(\log\sigma\log\log n)^{1+\eps})$ time \cite{HST12}; a recent
unpublished result \cite{NTisaac13} obtains time $O(p+k\log^* k)$.
This is very close to optimal, but not yet our $O(p/\log_\sigma n + k)$ time.

In practice, the most compact implementation in this trend \cite{NV12}
reaches about 1--2 times the text size (including a representation of the 
text) and retrieves each of the top-$k$
results within milliseconds. An implementation of the ideas we propose in
this article \cite{KN13} makes use of the fact that, under very general 
probabilistic models, the average height of the suffix tree (and hence of
our grids) is $O(\log n)$ \cite{Szp92}. This enables a simple implementation
of our grid-based index that uses 2--3 times the text size and, although its
average query time, $O(p + (k+\log \log n)\log \log n)$, is not optimal, it
returns each answer within microseconds \cite{KN13}. 

\paragraph{More complex queries.}
In the long term, the most interesting open questions are related to extending
the one-pattern results to the bag-of-words paradigm of information retrieval. 
Our model easily handles single-word searches, and also phrases (which is 
quite complicated with inverted indexes \cite{ZM06,BYRN11}, particularly if 
their 
weights have to be computed). Handling a set of words or phrases, whose weights
within any document $d$ must be combined in some form (for example using the 
$\tf \times \idf$ model) is more challenging. We are only aware of some very 
preliminary results for this case \cite{CP10,GNP11,HSTV10}, which suggest that 
it is unlikely that strong worst-case results can be obtained. Instead, one can
aim at complexities related to the results achieved with inverted lists on the
simpler natural language model. It is 
interesting to note that our online result allows simulating the left-to-right
traversal, in decreasing weight order, of the virtual list of occurrences of
any string pattern $P$. Therefore, for a bag-of-word queries, we can emulate
any algorithm designed for inverted indexes which stores those lists in
explicit form \cite{PZSD96,AM06}, therefore extending 
any such technique to the general model of string documents.

\subsection*{Acknowledgements}

We thank Djamal Belazzougui and Roberto Grossi for helpful pointers.

\bibliographystyle{plain}
\bibliography{paper}

\end{document}